%% file: csWithNonSquashingIntegers.tex
\input{variables}
\def\IEEEsubmission{0}
\def\totalLength{M}
\def\remaningLength{Z}
\def\remaningLengthAnother{Y}
\def\integerToBePartitioned{P}
\def\partition[#1]{p_{#1}}
\def\partition[#1]{p_{#1}}
\def\separationFreq[#1]{{s_{#1}}}	
\def\separationFreqFix{s'}
\def\separationFreqHelp[#1]{z_{#1}}	
	
\def\cardinality[#1][#2]{{\mathcal{A}}^{(#1)}\left(#2\right)}
\def\cardinalityP[#1][#2]{{\mathcal{B}}^{(#1)}\left(#2\right)}
\def\cardinalityDistance[#1][#2]{{\mathcal{B}}^{(#1)}\left(#2\right)}
\def\indexForCardinality{i}
\def\numberOfBitsOnPartitions{n_{\text{supp}}}
\def\numberOfBitsOnNZ{n_{\text{non-zero}}}
\def\numberOfBitsOnTotal{n_{\text{total}}}

\def\bitsForSep{\textit{\textbf{b}}_{{\rm \pi}, {\rm supp} }}

\def\bitsForSepD{{\hat{\textit{\textbf{b}}}}_{{\rm \pi}, {\rm supp} }}
\def\bitsForPhaseD{{\hat{\textit{\textbf{b}}}}_{\rm non\text{-}zero}}
\def\lowerBoundMinimumDistance{d_{\rm lb}}

\def\bitsForPhase{\textit{\textbf{b}}_{\rm non\text{-}zero}}
\def\decimalForSepAndPerm{i_{{\rm \pi}, {\rm supp}}}
\def\decimalForPerm{i_{{\rm \pi}}}
\def\decimalForSep{i_{{\rm supp}}}
\def\decimalForSepD{\hat{i}_{{\rm supp}}}
\def\decimalForPermD{\hat{i}_{{\rm \pi}}}
\def\decimalForSepAndPermD{\hat{i}_{{\rm \pi}, {\rm supp}}}

\def\minimumDistanceNZ{d_{\rm non\text{-}zero}}
\def\minimumDistancePart{d_{\rm supp}}
\def\minimumDistance{d_{\rm min}}
\def\magnitudePSK{r}
\def\spacing{l}
\def\operatorSupport[#1]{{\rm supp}({#1})}

\def\integerToBeMapped{n}
\def\algorithmEncoder[#1]{{{\epsilon}_{\rm 1}}(#1)}
\def\algorithmDecoder[#1]{{\epsilon}_{\rm 1}^{-1}(#1)}
\def\algorithmEncoderP[#1]{{{\epsilon}_{\rm 2}}(#1)}
\def\algorithmDecoderP[#1]{{\epsilon}_{\rm 2}^{-1}(#1)}
\def\algorithmEncoderDeltaSeparation[#1]{{\epsilon}_{\rm 3}(#1)}
\def\algorithmDecoderDeltaSeparation[#1]{{\epsilon}_{\rm 3}^{-1}(#1)}
\def\NsumAll{{n{(c)}}}
\def\NsumAllGiven[#1]{{n{(#1)}}}

\def\chosenReferencePointOther{{c}}

\def\algorithmDecoderFinal[#1]{\text{dec}(#1)}
\def\seqW[#1]{\textit{\textbf{w}}_{#1}}

\def\seqS{\textbf{\textit{s}}}
\def\seqStilde{{\hat{{\bm{s}}}}}

\def\referenceDerivative{\ell}

\def\receivedElement[#1]{r_{#1}}
\def\channel[#1]{c_{#1}}
\def\noise[#1]{n_{#1}}

\def\parametersImag{\theta}

\def\parametersImagg{{{\psi}}}

\def\parametersImagf{{{\phi}}}

\def\parametersImagEstimate{{\hat{\theta}}}

\def\weightedElement[#1]{w_{#1}}
\def\weightedElementP[#1]{v_{#1}}
\def\channelPowerElement[#1]{h_{#1}}
\def\weightedElementAfterSum[#1]{w_{#1}^{\referenceDerivative}}
\def\weightedElementAfterSumML[#1]{\hat{w}_{#1}^{\referenceDerivative}}
\def\channelPowerElementAfterSum[#1][#2]{h_{#1}^{#2}}

\def\funcfForFinalAmplitudeDerivate[#1][#2]{f_{\rm r}^{#1}(#2)}
\def\funcfForFinalPhaseDerivate[#1][#2]{f_{\rm i}^{#1}(#2)}
\def\funcfForFinalPhaseDerivateDec[#1][#2]{\check{f}_{\rm i}^{#1}(#2)}
\def\funcgForFinalAmplitudeDerivate[#1][#2]{g_{\rm r}^{#1}(#2)}
\def\funcgForFinalPhaseDerivate[#1][#2]{g_{\rm i}^{#1}(#2)}
\def\funcgForFinalPhaseDerivateDec[#1][#2]{\check{g}_{\rm i}^{#1}(#2)}
\def\mappingFunction[#1]{{M_{\referenceDerivative}(#1)}}

\def\weightedSequence[#1]{\textit{{\textbf{w}}}_{#1}}
\def\weightedSequenceSub[#1]{\textit{{\textbf{w}}}_{#1}'}
\def\weightedChannel[#1]{{{{q}}}_{#1}}
\def\weightedChannelSub[#1]{{{{q}}}_{#1}'}
\def\permSequence[#1]{{{{l}}}_{#1}}
\def\permSequenceSub[#1]{{{{l}}}_{#1}'}
\def\numberOfSequences{N}

\def\setOfweightedSequence{\textit{\textbf{W}}_{\numberOfIterations}}
\def\setOfweightedChannel{\textit{\textbf{q}}_{\numberOfIterations}}
\def\setOfpermutations{\textit{\textbf{l}}_{\numberOfIterations}}

\def\setOfweightedSequenceSub{\textit{\textbf{W}}_{\numberOfIterations-1}}
\def\setOfweightedChannelSub{\textit{\textbf{q}}_{\numberOfIterations-1}}
\def\setOfpermutationsSub{\textit{\textbf{l}}_{\numberOfIterations-1}}

\def\indexOptimum{\hat{n}_{\numberOfIterations}}
\def\indexOptimumPre{\hat{n}_{\numberOfIterations-1}}

\def\indexEnumaration{i}

\def\phaseOffsetl[#1]{\Delta_{#1}}
\def\phaseOffsetll[#1]{\Delta_{#1}}
\def\amplitudeOffset[#1]{\epsilon'_{#1}}
\def\amplitudeOffsetl[#1]{\epsilon_{#1}}

\def\numberOfGoodSeqences{{N_{\rm best}}}
\def\numberOfSelectSeqences{{N_{\rm select}}}
\def\numberOfSelectSeqencesMax{{N_{\rm max}}}

\def\setOfgoodSequenceIndexSub{\textit{\textbf{n}}_{\numberOfIterations-1}}
\def\EbNO{E_{\rm b}/N_{0}}

\def\indexRecursion{i}
\def\functionIndex[#1][#2]{I_{#1}(#2)}

\def\channelSum[#1]{q_{#1}}

\def\referencePoint[#1]{r_{#1}}
\def\roomForAReferencePoint[#1]{Z(#1)}

\newcommand\mydots{\hbox to 1em{.\hss.\hss.}}
\if\IEEEsubmission1
\documentclass[journal,12pt,onecolumn,draftclsnofoot,]{IEEEtran}
\else
\documentclass[journal]{IEEEtran}
\fi

\usepackage{multicol}
\usepackage{lipsum}
\usepackage{mathtools}
\usepackage{amsthm}
\usepackage{acronym}
\usepackage[bookmarksopen=true]{hyperref}
\usepackage{stfloats}
\usepackage{amsfonts}
\usepackage{acronym}
\usepackage{cite}
\usepackage{multirow}
\usepackage{bm}
\usepackage{amsmath,amssymb}
\usepackage{graphicx}
\usepackage[table,xcdraw]{xcolor}
\usepackage[caption=false,font=footnotesize]{subfig}
\usepackage[geometry]{ifsym}
\usepackage{array}
\usepackage{flushend}
\usepackage[utf8]{inputenc}
\usepackage[T1]{fontenc}
\usepackage{algpseudocode}

\usepackage[ruled,vlined]{algorithm2e}

\usepackage{tikz}
\usepackage{lipsum}
\usetikzlibrary{calc,shadings,patterns}
\makeatletter
\tikzset{%
  remember picture with id/.style={%
    remember picture,
    overlay,
    save picture id=#1,
  },
  save picture id/.code={%
    \edef\pgf@temp{#1}%
    \immediate\write\pgfutil@auxout{%
      \noexpand\savepointas{\pgf@temp}{\pgfpictureid}}%
  },
  if picture id/.code args={#1#2#3}{%
    \@ifundefined{save@pt@#1}{%
      \pgfkeysalso{#3}%
    }{
      \pgfkeysalso{#2}%
    }
  }
}

\def\savepointas#1#2{%
  \expandafter\gdef\csname save@pt@#1\endcsname{#2}%
}

\def\tmk@labeldef#1,#2\@nil{%
  \def\tmk@label{#1}%
  \def\tmk@def{#2}%
}

\tikzdeclarecoordinatesystem{pic}{%
  \pgfutil@in@,{#1}%
  \ifpgfutil@in@%
    \tmk@labeldef#1\@nil
  \else
    \tmk@labeldef#1,(0pt,0pt)\@nil
  \fi
  \@ifundefined{save@pt@\tmk@label}{%
    \tikz@scan@one@point\pgfutil@firstofone\tmk@def
  }{%
  \pgfsys@getposition{\csname save@pt@\tmk@label\endcsname}\save@orig@pic%
  \pgfsys@getposition{\pgfpictureid}\save@this@pic%
  \pgf@process{\pgfpointorigin\save@this@pic}%
  \pgf@xa=\pgf@x
  \pgf@ya=\pgf@y
  \pgf@process{\pgfpointorigin\save@orig@pic}%
  \advance\pgf@x by -\pgf@xa
  \advance\pgf@y by -\pgf@ya
  }%
}

\makeatother
\newcounter{hatchNumber}
\DeclarePairedDelimiter\floor{\lfloor}{\rfloor}

\usepackage{cuted}
\makeatletter
\newif\ifAC@uppercase@first%
\def\Aclp#1{\AC@uppercase@firsttrue\aclp{#1}\AC@uppercase@firstfalse}%
\def\AC@aclp#1{%
	\ifcsname fn@#1@PL\endcsname%
	\ifAC@uppercase@first%
	\expandafter\expandafter\expandafter\MakeUppercase\csname fn@#1@PL\endcsname%
	\else%
	\csname fn@#1@PL\endcsname%
	\fi%
	\else%
	\AC@acl{#1}s%
	\fi%
}%
\def\Acp#1{\AC@uppercase@firsttrue\acp{#1}\AC@uppercase@firstfalse}%
\def\AC@acp#1{%
	\ifcsname fn@#1@PL\endcsname%
	\ifAC@uppercase@first%
	\expandafter\expandafter\expandafter\MakeUppercase\csname fn@#1@PL\endcsname%
	\else%
	\csname fn@#1@PL\endcsname%
	\fi%
	\else%
	\AC@ac{#1}s%
	\fi%
}%
\def\Acfp#1{\AC@uppercase@firsttrue\acfp{#1}\AC@uppercase@firstfalse}%
\def\AC@acfp#1{%
	\ifcsname fn@#1@PL\endcsname%
	\ifAC@uppercase@first%
	\expandafter\expandafter\expandafter\MakeUppercase\csname fn@#1@PL\endcsname%
	\else%
	\csname fn@#1@PL\endcsname%
	\fi%
	\else%
	\AC@acf{#1}s%
	\fi%
}%
\def\Acsp#1{\AC@uppercase@firsttrue\acsp{#1}\AC@uppercase@firstfalse}%
\def\AC@acsp#1{%
	\ifcsname fn@#1@PL\endcsname%
	\ifAC@uppercase@first%
	\expandafter\expandafter\expandafter\MakeUppercase\csname fn@#1@PL\endcsname%
	\else%
	\csname fn@#1@PL\endcsname%
	\fi%
	\else%
	\AC@acs{#1}s%
	\fi%
}%
\edef\AC@uppercase@write{\string\ifAC@uppercase@first\string\expandafter\string\MakeUppercase\string\fi\space}%
\def\AC@acrodef#1[#2]#3{%
	\@bsphack%
	\protected@write\@auxout{}{%
		\string\newacro{#1}[#2]{\AC@uppercase@write #3}%
	}\@esphack%
}%
\def\Acl#1{\AC@uppercase@firsttrue\acl{#1}\AC@uppercase@firstfalse}
\def\Acf#1{\AC@uppercase@firsttrue\acf{#1}\AC@uppercase@firstfalse}
\def\Ac#1{\AC@uppercase@firsttrue\ac{#1}\AC@uppercase@firstfalse}
\def\Acs#1{\AC@uppercase@firsttrue\acs{#1}\AC@uppercase@firstfalse}
\newtheorem{theorem}{Theorem}
\newtheorem{definition}{Definition}
\newtheorem{lemma}{Lemma}

\newtheorem{corollary}{Corollary}
 
\newtheorem{example}{\color{black} Example} 

\acrodef{CRC}{cyclic redundancy check}
\acrodef{UL}{uplink}
\acrodef{SIC}{successive interference cancellation}
\acrodef{PAPR}{peak-to-average-power ratio}
\acrodef{PMEPR}{peak-to-mean-envelope-power ratio}
\acrodef{AACF}{aperiodic auto-correlation function}
\acrodef{OFDM}{orthogonal frequency division multiplexing}
\acrodef{DFT}{discrete Fourier transform}
\acrodef{DC}{direct current}
\acrodef{CS}{complementary sequence}
\acrodef{GCP}{Golay complementary pair}
\acrodef{ANF}{algebraic normal form}
\acrodef{PSK}{phase shift keying}
\acrodef{QAM}{quadrature amplitude modulation}
\acrodef{QPSK}{quadrature phase-shift keying}
\acrodef{GDJ}{Golay-Davis-Jedwab}
\acrodef{FFT}{fast Fourier transform}
\acrodef{BER}{bit-error ratio}
\acrodef{SNR}{signal-to-noise ratio}
\acrodef{4G}{Fourth Generation}
\acrodef{5G}{Fifth Generation}
\acrodef{NR}{New Radio}
\acrodef{LTE}{Long-Term Evolution}
\acrodef{PTS}{partial transmit sequences}
\acrodef{PSD}{power spectral density}
\acrodef{LDPC}{low-density parity check}
\acrodef{SE}{spectral efficiency}
\acrodef{eLAA}{enhanced licensed-assisted access}
\acrodef{NR-U}{NR operation in unlicensed bands}
\acrodef{RM}{Reed-Muller}
\acrodef{AE}{autoencoder}
\acrodef{DNN}{deep neural network}
\acrodef{OFDM-AE}{OFDM-based autoencoder}
\acrodef{DL}{deep learning}
\acrodef{CP}{cyclic prefix}
\acrodef{AWGN}{additive white Gaussian noise}
\acrodef{P2C}{polar-to-Cartesian}
\acrodef{CFR}{channel frequency response}
\acrodef{ReLU}{rectified linear unit}
\acrodef{MMSE}{minimum mean sqaure error}
\acrodef{BPSK}{binary phase-shift keying}
\acrodef{BLER}{block error rate}
\acrodef{ML}{maximum-likelihood}
\acrodef{PHY}{physical layer}
\acrodef{PA}{power amplifier}
\acrodef{IDFT}{inverse DFT}
\acrodef{DoF}{degrees-of-freedom}
\acrodef{IoT}{Internet-of-Things}
\acrodef{DFT-s-OFDM}{\ac{DFT}-spread \ac{OFDM}}
\acrodef{MMSE}{minimum mean square error}
\acrodef{FDE}{frequency-domain equalization}
\acrodef{FrFT}{fractional Fourier transform}
\acrodef{TF}{time-frequency}
\acrodef{BFSK}{binary frequency-shift keying}
\acrodef{CSS}{chirp spread spectrum}
\acrodef{BCSS}{binary chirp spread spectrum}
\acrodef{EVA}{Extended Vehicular A}
\acrodef{MIMO}{multi-input multi-output}
\acrodef{PIC}{parallel interference cancellation}
\acrodef{LoRa}{Long Range}
\acrodef{HF}{high-frequency}
\acrodef{FDSS}{frequency-domain spectral shaping}
\acrodef{CSC}{circularly-shifted chirp}
\acrodef{ISI}{inter-symbol interference}
\acrodef{DFRC}{dual-function radar and communication}
\acrodef{IM}{index modulation}
\acrodef{OFDM-IM}{OFDM with \ac{IM}}
\acrodef{NBN}{non-branching node}
\acrodef{CM}{composition modulation}
\acrodef{WCM}{weak composition modulation}
\acrodef{ICM}{index and composition modulation}
\acrodef{SPM}{set partition modulation}
\acrodef{SC}{single-carrier}
\acrodef{SC-IM}{single-carrier waveform with IM}
\acrodef{CSC-IM}{circularly-shifted chirps with IM}
\acrodef{OOK}{on-off keying}
\acrodef{AI}{artificial intelligence}
\acrodef{NN}{neural network}
\acrodef{CSS}{complementary sequence set}
\begin{document}
\title{ 
Encoding and Decoding with Partitioned Complementary Sequences for Low-PAPR OFDM
}
\author{Alphan~\c{S}ahin
\thanks{The author is with the University of South Carolina, Columbia, SC. E-mail: asahin@mailbox.sc.edu.}
%\author{Author 1, Author 2, Author 3, and Author 4 }
}

%\author{
%	\IEEEauthorblockN{Alphan~\c{S}ahin, Nozhan~Hosseini, Hosseinali~Jamal, and David~W.~Matolak}
%	\IEEEauthorblockA{Electrical  Engineering Department,
%		University of South Carolina, Columbia, SC, USA}
%	Email:  asahin@mailbox.sc.edu, nozhan@email.sc.edu, hjamal@email.sc.edu, matolak@cec.sc.edu}

\maketitle

\begin{abstract}
In this study, we propose partitioned \acp{CS} where the gaps between the clusters encode information bits to achieve low \ac{PAPR} \ac{OFDM} symbols. 
We show that the partitioning rule without losing the feature of being a \ac{CS} coincides with the non-squashing partitions of a positive integer and leads to a symmetric separation of clusters.
We analytically derive the number of partitioned \acp{CS} for  given bandwidth and a minimum distance constraint and obtain the corresponding recursive  methods for enumerating the values of separations. We show that partitioning can increase the \ac{SE} without changing the alphabet of the non-zero elements of the \ac{CS}, i.e., standard \acp{CS} relying on \ac{RM} code.
We also develop an encoder for partitioned \acp{CS} and a maximum-likelihood-based recursive decoder for \ac{AWGN} and fading channels. Our results indicate that the partitioned \acp{CS}  under a minimum distance constraint can perform similar to the standard \acp{CS} in terms of average \ac{BLER} and provide a higher \ac{SE} at the expense of a limited \ac{SNR} loss.
\end{abstract}
\acresetall

\section{Introduction}

High \ac{PAPR}  is a long-lasting problem of an \ac{OFDM} transmission. Among many other \ac{PAPR} mitigation methods \cite{Wunder_2013, Rahmatallah_2013}, \acp{CS}, introduced by Marcel Golay \cite{Golay_1961}, allow one to limit the peak instantaneous power of \ac{OFDM} signals without any optimization method \cite{Popovic_1991}. 
In \cite{davis_1999}, Davis and Jedwab showed that $\numberOfIterations!/2\cdot2^{h(\numberOfIterations+1)}$ \acp{CS} of length $2^\numberOfIterations$  for $h\in\integersPositive$ occur as the elements of the cosets of the first-order \ac{RM} code within the second-order \ac{RM} code. Hence, they obtained a notable coding scheme guaranteeing a low \ac{PAPR} for \ac{OFDM} symbols while providing good error correction capability.  The set of \acp{CS} based on Davis and Jedwab's construction is often referred to as {\em \ac{GDJ}} sequences or {\em standard} \acp{CS}.

One drawback of the code proposed in \cite{davis_1999} is the low \ac{SE}. To address this issue, one direction is to obtain \acp{CS} that cannot be generated through the method in \cite{davis_1999}, i.e., {\em non-standard} sequences. In \cite{Li_2005}, it was shown that some Boolean functions containing third-order monomials can lead to \acp{CS}. Another direction is to synthesize \acp{CS} where their elements belong to a larger alphabet such as \ac{QAM} constellation. Li showed that there exist at least $[(\numberOfIterations+1)4^{2(q-1)}-(m+1)4^{(q-1)}+2^{q-1}](m!/2)4^{(m+1)}$ \acp{CS} with $4^q$-QAM alphabet \cite{Li_2010}, which generalizes the results in earlier work in \cite{robing_2001,Chong_2002,Chong_2003,Lee_2006,Chang_2010,Li_2008} by using a method called {\em offset method}. In \cite{budisin_2018}, \ac{QAM} \acp{CS} were synthesized by indexing the elements of the unitary matrices and using the properties of Gaussian integers. 
Another approach that has recently attracted significant interest is to construct \acp{CSS} or complete complementary code as discussed in \cite{Chen_2018, Das_2020, Wu_2020} and the references therein, which relax the maximum peak instantaneous power of \acp{CS}, but are useful to obtain many sequences with large zero-correlation zones.
Although the aforementioned constructions are remarkable, they often do not reveal the encoding and decoding procedures.
In \cite{sahin_2020gm}, a \ac{CS} construction that utilizes different pseudo-Boolean functions for the magnitude and phase of the elements of the synthesized \ac{CS}, the support, and  the seed \ac{GCP} was proposed. Since this construction generalizes Davis and Jedwab's method through independently configurable functions, it allows one to  develop an encoder and a decoder for \acp{CS} while enabling  \acp{CS} with zero-valued elements and \acp{CS} with uniform and non-uniform constellations. This construction was used to develop a neural-network-based encoder and a decoder in \cite{sahin_ICC2020} and a low-\ac{PAPR} multi-user scheme in the uplink for the interlaced allocation in 3GPP \ac{5G} \ac{NR-U} \cite{sahin_TWC2020}.

Developing a low-\ac{PAPR} encoder based on \acp{CS} is not  a straightforward task since a set of different \acp{CS} is typically constrained in terms of size, sequence length, and alphabet. For example, \acp{CS} with a high-order modulation can alter the mean \ac{OFDM} symbol power and cause a \ac{PAPR} larger than $3$~dB when the entire transmission is considered for the average power calculation. Although it is possible to address this issue by constraining the magnitude of the elements of \acp{CS} as done in \cite{sahin_2020gm}, this issue increases the design complexity  and decreases the number of different \acp{CS}. The performance of \acp{CS} with a high-order modulation can also be worse than the one for \acp{CS} with a \ac{PSK} alphabet, e.g., standard sequences, since a high-order modulation often decreases the minimum Euclidean distance of the set of \acp{CS}. It is also challenging to enumerate \acp{CS} with an arbitrary length. For example, if the  \ac{DoF} are different from a typical \ac{CS} length, i.e., $2^\numberOfIterations$, the available \ac{DoF} are not fully exploited to increase the data rate or reduce the error rate. In this study, we aim to address these problems by partitioning the standard \acp{CS} based on the information bits by using the theoretical framework in \cite{sahin_2020gm}, and use \acp{CS} with zero-valued elements. It is worth noting that the \acp{CS} with zero-valued elements are known in the literature. However, they are primarily used to address resource allocation in an \ac{OFDM} symbol as in \cite{sahin_2020gm,sahin_TWC2020}, and \cite{Sahin_2019ba}. In \cite{Zhou_CSS_2018}, the  \acp{CSS} with zero elements are generated through an iterative method with the motivation of cognitive radio applications. To the best of our knowledge, the systematic design of the zero-valued elements of a \ac{CS} to transmit extra information bits   has not been discussed in the literature.

%for extra information transmission and its systematic design to address the issue of high \ac{PAPR} of \ac{OFDM} symbols have not been discussed in the literature. 

%A non-contiguous resource allocation in an \ac{OFDM} symbol can be beneficial for low-latency and reliable communications or meeting regulatory  requirements for the communications in unlicensed bands, e.g., interlaced allocation in 3GPP \ac{5G} \ac{NR-U} \cite{ericUL}.  However, when some of the OFDM subcarriers within the band are not utilized, the \ac{PAPR} minimization can be a challenging task due to the increased number of conditions for achieving a high merit factor for the sequences with zeroes elements \cite{sahin_TWC2020, Liu_2018}. In \cite{sahin_TWC2020}, this challenge is addressed by using the properties of \acp{CS} for a given resource allocation in \ac{OFDM} symbols with non-contiguous resource allocation for 3GPP \ac{5G} \ac{NR-U}  by exploiting a generalized Golay's concatenation and interleaving methods \cite{Golay_1961} and the framework introduced in \cite{sahin_2020gm}. In \cite{Gokceli_2019}, an iterative clipping and error filtering based on the flexible non-contiguous resource allocation in \ac{5G} \ac{NR} is proposed for reducing \ac{PAPR}. In \cite{Sahin_2019ba}, the authors also exploited the \acp{CS} for designing low \ac{PAPR} multi-channel \ac{OOK} waveform for frequency division multiplexing in IEEE 802.11ba Wake-up Radio, which is also a non-contiguous resource allocation. In these studies, the zero-valued subcarriers are utilized for information transmission.

The proposed concept is inherently related to the \ac{IM}.
 %since it conveys information by altering the zero-valued subcarriers of \ac{OFDM}.
 \Ac{IM} is a subclass of permutation modulation \cite{Slepian_1965} and allows one to encode information in the order of discrete objects. 
 %It has been considered for many areas of communication systems. 
Turning on and off the antennas for transmitting extra information, i.e., spatial modulation \cite{Mesleh_2008}, adjusting the on/off status of available RF mirrors and encoding information on the antenna pattern, i.e., media-based modulation \cite{Khandani_2013},  activating/deactivating \ac{OFDM} subcarriers with modulation symbols, i.e., \ac{OFDM-IM} \cite{basar_2013}, are a few examples of the many variants of \ac{IM}. In \cite{Jaradat_2020}, extra information bits were transmitted by exploiting the gap between the active subcarriers in each sub-block of \ac{OFDM-IM}. However, the gaps are not tied to a special structure and associated encoding/decoding operations based on structured gaps are not discussed. Since \ac{OFDM-IM} is also based on \ac{OFDM}, it also suffers from high \ac{PAPR}. In \cite{Vora_2018}, circular-time shifts are applied to \ac{OFDM-IM} to reduce the \ac{PAPR}  for multiple antennas. In \cite{Kim_2019} and \cite{Zheng_2017}, the dither signals are considered to reduce the \ac{PAPR} of \ac{OFDM-IM} symbols by using convex optimization techniques. In \cite{Sugiura_2017} and \cite{Nakao_2017}, a \ac{SC-IM} is investigated for achieving a low \ac{PAPR} transmission, which eliminates data-dependent optimization techniques. It was shown that the \ac{IM}  can slightly degrade the \ac{PAPR} benefit of a typical \ac{SC} transmission with low-order modulations. This is expected as \ac{IM} increases the zero crossings in the time domain for \ac{SC}. In the extreme cases where there is a large \ac{DoF} for indices with a few index choices, the \ac{PAPR} of \ac{SC-IM} symbols can even be worse than that of \ac{OFDM-IM} due to pulses in an \ac{SC} scheme \cite{Safi_2020_CCNC}. %\cite{Safi_2020_CCNC} also showed that \ac{CSC-IM} can provide a lower \ac{PAPR} than \ac{SC-IM}  for low-data rates and chirps can be utilized to synthesize \acp{CS}. 
To the best of our knowledge, \acp{CS} have not been systematically investigated from the \ac{IM} perspective  in the literature.

\subsection{Contributions and Organization}

{\textbf{Comprehensive analysis of partitioned CSs:} }
We analytically derive the number of partitioned \acp{CS} for a given bandwidth. % Our derivation relies on the symmetric nature of  the partitioning rule to be a \ac{CS}, which is also related to non-squashing partitions of a positive integer.
 We show that the number of \acp{CS} increases by a large factor when  the partitioning is taken into account and  the alphabet of the non-zero elements of the \ac{CS} remains the same. We also derive the algorithms that map a natural number to the separations between the clusters or vice versa, which are needed for an encoder and a decoder based on partitioned \acp{CS}.

{\textbf{Partitioning under a minimum distance constraint:} }  To obtain the partitioned \acp{CS} under a minimum Euclidean distance constraint, we propose a partitioning strategy based on the symmetric structure of the partitioned \acp{CS}. For a given minimum distance, we derive the cardinality of the partitioned \acp{CS} and the algorithms that construct a bijective mapping between a natural number and the separations between the clusters. We show that partitioned \acp{CS} can maintain the distance properties of the standard \acp{CS} while offering a similar \ac{SE} and a better flexibility in bandwidth.

{\textbf{Encoder/decoder for partitioned \acp{CS}:} } We develop an encoder and a  recursive decoder for partitioned \acp{CS} and compare the partitioned \acp{CS} with the standard \acp{CS} and polar code in various configurations. 

The rest of the paper is organized as follows. In Section \ref{sec:preliminaries}, preliminary discussions on \acp{CS} and notations are provided. In Section \ref{sec:CSwithNS}, partitioning for \acp{CS} is investigated. The cardinality of the partitioned \acp{CS} under a minimum distance constraint and the algorithms for enumerating the partitions are discussed.  In Section \ref{sec:encAndDec}, we use the partitioned \acp{CS} to obtain low-\ac{PAPR} \ac{OFDM} symbols and present the associated encoder and decoder. In Section~\ref{sec:numeric}, we evaluate the encoder and decoder and compare partitioned \acp{CS} with the standard \acp{CS} and polar code. We conclude the paper in Section \ref{sec:conc} with final remarks.

{\em Notation:} The sets of complex numbers, real numbers,  integers,  non-negative integers,  positive integers (i.e., natural numbers), and integers modulo $\numberOfPointsForPSK$ are denoted by $\complexNumbers$,  $\realNumbers$,  $\integers$, $\integersNonnegative$, $\integersPositive$, and $\integers_\numberOfPointsForPSK$, respectively. The set of $\numberOfIterations$-dimensional integers where each element is in $\integers_\numberOfPointsForPSK$  is denoted by $\integers^\numberOfIterations_\numberOfPointsForPSK$.  
The assignment operation is denoted by $\leftarrow$.  The constant $\constantj$ is $\sqrt{-1}$.
% An array of $L$ sequences is denoted as ${\textit{\textbf{S}}}=(\seqGa_i)_{i=0}^{L-1}$. %The operation $\seqGa\pm\seqGb$ and the operation $\seqGa \odot \seqGb$ are the element-wise summation/subtraction and the element-wise multiplication of the sequence $\seqGa$ and the sequence $\seqGb$, respectively.
%Let $(i_1,\mydots,i_{\lengthGaGb})$ be a permutation of $(1,\mydots,\lengthGaGb)$. The sequence  $(\eleGa[{0}],\mydots,\eleGa[{\lengthGaGb-1}])_{(i_0,\mydots,i_{\lengthGaGb-1})}$ is equal to $(\eleGa[{i_0}-1],\mydots,\eleGa[i_{\lengthGaGb-1}-1])$, which is a permutation of $(\eleGa[{0}],\mydots,\eleGa[{\lengthGaGb-1}])$.

\section{Preliminaries and Definitions}
\label{sec:preliminaries}

\subsection{Sequences, OFDM, and PMEPR}
Let $\seqGa=(\eleGa[{\indexEleOfSeq}])_{i=0}^{\lengthOfSequence-1}\triangleq(\eleGa[0],\eleGa[1],\dots, \eleGa[\lengthOfSequence-1])$ be a sequence  of length $\lengthOfSequence$,  where $\eleGa[{\indexEleOfSeq}]\in\complexNumbers$ and $\eleGa[\lengthOfSequence-1]\neq0$. We associate the sequence $\seqGa$ with the polynomial
$\polySeq[\seqGaP][\polyVariable] = \eleGa[\lengthOfSequence-1]\polyVariable^{\lengthOfSequence-1} + \eleGa[\lengthOfSequence-2]\polyVariable^{\lengthOfSequence-2}+ \dots + \eleGa[0]$
in indeterminate $\polyVariable$. The \ac{AACF} of the sequence $\seqGa$ is defined as
\begin{align}
	\apac[\seqGa][\lagForCorrelation]\triangleq
	\begin{cases}
		\sum_{\indexEleOfSeq=0}^{\lengthOfSequence-\lagForCorrelation-1} \eleGa[\indexEleOfSeq]^*\eleGa[\indexEleOfSeq+\lagForCorrelation], & 0\le\lagForCorrelation\le\lengthOfSequence-1\\
		\sum_{\indexEleOfSeq=0}^{\lengthOfSequence+\lagForCorrelation-1} \eleGa[\indexEleOfSeq]\eleGa[\indexEleOfSeq-\lagForCorrelation]^*, & -\lengthOfSequence+1\le\lagForCorrelation<0\\
		0,& \text{otherwise}
	\end{cases}~.
\end{align}

We express the  continuous-time baseband \ac{OFDM} symbol generated from the sequence $\seqGa$  as $\OFDMinTime[\seqGa][\timeVar]=\sum_{\indexEleOfSeq=0}^{\lengthOfSequence-1}\eleGa[{\indexEleOfSeq}]\constante^{\constantj2\pi\indexEleOfSeq\frac{\timeVar}{\symbolDuration}}$ for $\timeVar\in[0,\symbolDuration)$, where $\symbolDuration$ is the symbol duration. The {instantaneous envelope power} of the baseband \ac{OFDM} symbol can be examined by evaluating the polynomial $\polySeq[\seqGaP][\polyVariable]$ at $|\polyVariable|=1$  since $\OFDMinTime[\seqGa][\timeVar] = \polySeq[\seqGaP][{\constante^{\constantj\frac{2\pi\timeVar}{\symbolDuration}}}]$. 

In this study, the \ac{PMEPR} of an admissible sequence  $\seqGc\in\complexNumbers^{\lengthOfSequence}$ in a code is defined as
$\max_{\timeVar\in[0, \symbolDuration)}|\OFDMinTime[\seqGc][\timeVar]|^2/{\Ptransmit}
$, where $\Ptransmit=\expectationOperator[{\apac[\seqGc][0]}][\seqGc]$ is a constant that depends on the code. Note that  \ac{PAPR} measures the peakness of the signal at the analog front-end as $\max_{\timeVar\in[0, \symbolDuration)}\Re\{\OFDMinTime[\seqGc][\timeVar]\constante^{2\pi\carrierFrequency\timeVar}\}^2/{\Ptransmit}
$, where $\carrierFrequency$ is the carrier frequency  \cite{Sharif_2003}. As the \ac{PMEPR} is an upper bound for the \ac{PAPR} and its calculation is easier than \ac{PAPR}, we use \ac{PMEPR} to quantify the peakness of the signals in this study.  
\subsection{Sequence Synthesis}
\label{subsec:synthesis}
Let $\funcfForANF$ be a function that maps from $\integers^\numberOfIterations_2=\{(\monomial[1],\monomial[2],\dots, \monomial[\numberOfIterations])| \monomial[\indexFirstOrderMonomial]\in\integers_2\}$ to $\integers_\numberOfPointsForPSK$, where $\numberOfPointsForPSK$ is an integer. The function $\funcfForANF$ can be expressed as a linear combination of the monomials over $\integers_\numberOfPointsForPSK$, i.e.,
\begin{align}
	\hspace{-2mm}\funcfForANF(\seqx)= \sum_{\indexMonomial=0}^{2^\numberOfIterations-1} \coeffientsANF[\indexMonomial]\prod_{\indexFirstOrderMonomial=1}^{\numberOfIterations} \monomial[\indexFirstOrderMonomial]^{\orderMonomial[\indexFirstOrderMonomial]} = \coeffientsANF[0]1+ \dots+ \coeffientsANF[2^\numberOfIterations-1]\monomial[1]\monomial[2]\mydots\monomial[\numberOfIterations]~,
	\label{eq:ANF}
\end{align} 
where $\seqx\triangleq(\monomial[1],\monomial[2],\dots, \monomial[\numberOfIterations])$ and $\coeffientsANF[\indexMonomial]\in \integers_\numberOfPointsForPSK$ for $\indexMonomial = \sum_{\indexFirstOrderMonomial=1}^{\numberOfIterations}\orderMonomial[\indexFirstOrderMonomial]2^{\numberOfIterations-\indexFirstOrderMonomial}$ and $\orderMonomial[\indexFirstOrderMonomial]\in\integers_2$ (i.e., the coefficient of ($\indexMonomial+1$)th monomial $\monomial[1]^{\orderMonomial[1]}\monomial[2]^{\orderMonomial[2]}\mydots\monomial[\numberOfIterations]^{\orderMonomial[\numberOfIterations]}$  belongs to $\integers_\numberOfPointsForPSK$). If the monomial coefficients are in $\integersNonnegative$, the co-domain of  the function $\funcfForANF$ is $\integersNonnegative$. In this study, $\funcfForANFdec$ is defined as  $\funcfForANFdec(\varMonomial)\triangleq\funcfForANF\circ\funcEnumInv(\varMonomial)$, where $\varMonomial=\funcEnum(\seqx) \triangleq   \sum_{\indexFirstOrderMonomial=1}^{\numberOfIterations}\monomial[\indexFirstOrderMonomial]2^{\numberOfIterations-\indexFirstOrderMonomial}$. In other words, $\funcfForANFdec(\varMonomial)=\funcfForANF(\seqx)$, where $\varMonomial$ is the decimal representation of the binary number constructed using all elements in the sequence $\seqx$, where $\monomial[1]$ is the most significant bit.

Let $\funcfForFinalPhase:\integers^\numberOfIterations_2\rightarrow\integers_\numberOfPointsForPSK$ and $\funcfForCommonShift:\integers^\numberOfIterations_2\rightarrow\integersNonnegative$. We synthesize a unimodular complex sequence $\seqGa=(\eleGa[{\indexEleOfSeq}])_{\varMonomial=0}^{2^\numberOfIterations-1}$ by listing the values of $\eleGa[\varMonomial]=\exponentialBase^{\constantj \funcfForFinalPhaseDec(\varMonomial)}$ for $\exponentialBase\triangleq\constante^{\frac{2\pi}{\numberOfPointsForPSK}}$. By modifying the polynomial
associated with the sequence $\seqGa$, we generate a new sequence $\seqGc=(\eleGc[{\indexEleOfSeq}])_{\varMonomial=0}^{\lengthOfSequence-1}$ with zero elements as
\begin{align}
	\polySeq[\seqGcP][\polyVariable] =\sum_{\varMonomial=0}^{\lengthOfSequence-1} \eleGc[\varMonomial]\polyVariable^{\varMonomial}=\sum_{\varMonomial=0}^{2^\numberOfIterations-1} \eleGa[\varMonomial]\polyVariable^{\funcfForCommonShiftDec(\varMonomial)+\varMonomial}\nonumber~,
\end{align}
for $\lengthOfSequence\ge2^\numberOfIterations$. 

In this study, we define the support of the sequence $\seqGc$ as $\operatorSupport[{\seqGc}]\triangleq\{\varMonomial\in \integers_\lengthOfSequence|\eleGc[\varMonomial]\neq0\}$.

\subsection{Complementary Sequences}
The sequence pair  $(\seqGa,\seqGb)$ is  a \ac{GCP} if $\apac[\seqGa][\lagForCorrelation]+\apac[\seqGb][\lagForCorrelation] = 0$ for $\lagForCorrelation\neq0$ and
the sequences $\seqGa$ and $\seqGb$ are referred to as \acp{CS}. In this study, we  use a specific \ac{CS} construction that leads to \acp{CS} with zero elements, based on Theorem 2 discussed in \cite{sahin_2020gm}:
\begin{theorem}[\cite{sahin_2020gm}]
	\label{th:reduced}
	Let 
	$\seqPermutationCompShift=(\permutationMono[\indexIteration])_{\indexIteration=1}^{\numberOfIterations}$ be a sequence defined by a permutation of $\{1,2,\dots,\numberOfIterations\}$. For any $\separationGolayFix,\separationGolay[\indexIteration]\in \integersNonnegative$, $\numberOfPointsForPSK\in\integersPositive$, and $\angleexpAll[\indexIteration],\arbitraryPhaseK \in \integers_\numberOfPointsForPSK$ for $\indexIteration=1,2,\mydots,\numberOfIterations$, let
	\begin{align}
		&\funcfForFinalPhase(\seqx)
		= {\frac{\numberOfPointsForPSK}{2}\sum_{\indexIteration=1}^{\numberOfIterations-1}\monomial[{\permutationMono[{\indexIteration}]}]\monomial[{\permutationMono[{\indexIteration+1}]}]}+\sum_{\indexIteration=1}^\numberOfIterations \angleexpAll[\indexIteration]\monomial[{\permutationMono[{\indexIteration}]}]+  \arbitraryPhaseK\label{eq:imagPartReduced}~,
		\\	
		&\funcfForCommonShift(\seqx) =\sum_{\indexIteration=1}^\numberOfIterations\separationGolay[\indexIteration]\monomial[{{\indexIteration}}]+\separationGolayFix~.
		\label{eq:shift}
	\end{align}
	Then, the sequence $\seqGt=(\eleGt[\varMonomial])_{\varMonomial=0}^{\lengthOfSequence-1}$, where  its associated polynomial  is  given by
	\begin{align}
		\polySeq[{\seqGtP}][\polyVariable] &= 
		\sum_{\varMonomial=0}^{2^\numberOfIterations-1}
		%\funcfForCommonOrder(\seqx,\polyVariable)\times
		\exponentialBase^{\constantj \funcfForFinalPhaseDec(\varMonomial)}
		\polyVariable^{\funcfForCommonShiftDec(\varMonomial)+\varMonomial}~,\label{eq:encodedFOFDMonly}
	\end{align}
	is a \ac{CS} of length $\lengthOfSequence=2^\numberOfIterations+\separationGolayFix+\sum_{\indexIteration=1}^\numberOfIterations\separationGolay[\indexIteration]$.
\end{theorem}
For Theorem~\ref{th:reduced}, we set the polynomials related to the magnitude of the elements of the \ac{CS} and the seed \ac{GCP} in  \cite[Theorem~2]{sahin_2020gm} to $1$. %The permutations on the first-order monomials in \eqref{eq:imagPartReduced} and \eqref{eq:shift} are removed without loss of their generality. 
In addition, we introduce $\separationGolayFix$ that prepends $\separationGolayFix$ zero-valued elements to the synthesized sequence without changing the properties of \acp{CS}. We also  set the domain of $\arbitraryPhaseK$ and $\angleexpAll[\indexIteration]$ to $\integers_\numberOfPointsForPSK$ for $\indexIteration\in\{1,2,\mydots,\numberOfIterations\}$, i.e., \eqref{eq:imagPartReduced} leads to a second-order \ac{RM} code with $\numberOfPointsForPSK$-\ac{PSK} alphabet   \cite{davis_1999}.

If ${\funcfForCommonShiftDec(\varMonomial)+\varMonomial}$ and ${\funcfForCommonShiftDec(j)+j}$ are identical  for  $\varMonomial\neq j$, $\exponentialBase^{\constantj \funcfForFinalPhaseDec(\varMonomial)}$ and $\exponentialBase^{\constantj \funcfForFinalPhaseDec(j)}$ are superposed on the $({\funcfForCommonShiftDec(\varMonomial)+\varMonomial})$th element of the sequence. In  \cite{sahin_2020gm}, the conditions given by
\begin{align}
	\separationGolay[\ell] \ge	\separationGolay[\ell+1]+\separationGolay[\ell+2]+\cdots+ \separationGolay[\numberOfIterations] ~, 
	\label{eq:conditionNoOverlapping}
\end{align}
for  $1\le \ell\le \numberOfIterations-1$ are proposed to avoid such superpositions. In this study, a \ac{CS} is called a {\em non-partitioned \ac{CS}} if $\separationGolay[\indexIteration]=0$ for $\indexIteration\in\{1,2,\mydots,\numberOfIterations\}$. The set of non-partitioned $\numberOfPointsForPSK$-\ac{PSK} \acp{CS} based on \eqref{eq:imagPartReduced}  is referred to as {\em standard \acp{CS}}. If the conditions in \eqref{eq:conditionNoOverlapping} hold,  the resulting sequences  for any $\separationGolay[1],\mydots,\separationGolay[\indexIteration]\ge0$  are {\em partitioned \acp{CS}}.

\subsection{Non-squashing Partitions}
\label{subsec:nonsquashing}
A partition of a positive integer $\integerToBePartitioned$ into $\numberOfIterations$ parts, i.e., $\integerToBePartitioned=\partition[1]+\partition[2]+\cdots+\partition[\numberOfIterations]$, is defined as {\em non-squashing} \cite{sloane2005} if $\partition[\ell]\ge\partition[\ell+1]+\partition[\ell+2]+\cdots+\partition[\numberOfIterations]$  holds for $1\le\ell\le\numberOfIterations-1$ and $\partition[\indexIteration]\in\integersPositive$ for $\indexIteration=1,2,\mydots,\numberOfIterations$. 
The term of non-squashing partition was first coined in \cite{sloane2005} for  a problem called {\em box-stacking problem}. Assume that there are $\numberOfIterations$ boxes where each box is labeled as $1,2,\mydots,\numberOfIterations$. Suppose that the $\indexIteration$th
box can support a total weight of $\indexIteration$ grams. How many different ways of putting the boxes in a single stack are there such that no box will be squashed by the weight of the boxes above it? The solution to this problem and a bijective mapping between binary partitions \cite{Hirschhorn_2004}  and non-squashing partitions  were discussed in \cite{sloane2005}, which also led to various generalizations on non-squashing partitions \cite{sun_2018, Folsom_2016,Andrews_2007}.

The condition given in \eqref{eq:conditionNoOverlapping} to eliminate the superposition of elements of a \ac{CS} interestingly coincides with the non-squashing partitions of $\integerToBePartitioned=\sum_{\indexIteration=1}^\numberOfIterations\separationGolay[\indexIteration]$ into $\numberOfIterations$ parts, except that $\separationGolay[\indexIteration]$ can be equal to zero. Hence, $(\separationGolay[1],\separationGolay[2],\mydots,\separationGolay[\numberOfIterations])$ can be considered as a sequence 
where $(\separationGolay[1],\separationGolay[2],\mydots,\separationGolay[k])$ is a non-squashing partition of $\integerToBePartitioned=\sum_{\indexIteration=1}^k\separationGolay[\indexIteration]\ge1$ into  $k$ parts, $\separationGolay[\indexIteration]\neq0$ for $\indexIteration\in\{1,\mydots,k\}$ and $\separationGolay[k+1]=\cdots=\separationGolay[\numberOfIterations]=0$ for a given $k\in\{1,\mydots,\numberOfIterations\}$. For $k=0$, $\separationGolay[\indexIteration]=0$ for $\indexIteration\in\{1,2,\mydots,\numberOfIterations\}$.

\section{Partitioned Complementary Sequences}
\label{sec:CSwithNS}
\subsection{Problem Statement}
Consider a baseband \ac{OFDM} signal $\OFDMinTime[\seqGt][\timeVar]$, where the associated polynomial for the sequence $\seqGt$ is \eqref{eq:encodedFOFDMonly}.  Let $\totalLength$ be the number of available subcarriers  for $\totalLength\ge\lengthOfSequence\ge2^\numberOfIterations$. Assume that the conditions in \eqref{eq:conditionNoOverlapping} hold true. In this case, the number of zero-valued subcarriers in the frequency domain can be quantified as $\remaningLength=\totalLength-2^{\numberOfIterations}$ since
the number of non-zero elements of the sequence $\seqGt$ is $2^{\numberOfIterations}$ in Theorem~\ref{th:reduced}.
In this study, we  exploit the  $\remaningLength$ zero-valued subcarriers for encoding extra information bits by manipulating the support of the sequence $\seqGt$ with \eqref{eq:shift}, where the non-zero elements of sequence $\seqGt$ are from the standard \acp{CS} formed by \eqref{eq:imagPartReduced}. 

\label{subsec:problem}

\begin{figure}[t]
	\centering
	{\includegraphics[width =3.5in]{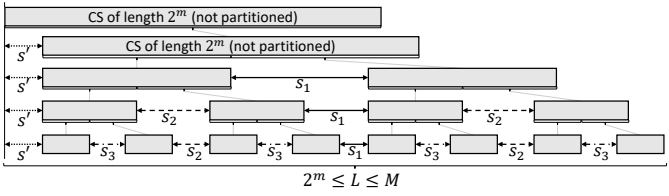}
	}
	\caption{The support of the CS changes based on the values of $\separationFreqFix$ and $\separationFreq[\indexIteration]$ for $\indexIteration=\{1,2,3\}$, where the number of zero-valued elements is $\separationFreqFix+\separationFreq[1]+2\separationFreq[2]+4\separationFreq[3]$. The number of zeroes can be chosen arbitrarily without affecting the properties of a \ac{CS} since the partitioning follows \eqref{eq:shift}.}
	\label{fig:partitioning}
\end{figure}
In Theorem~\ref{th:reduced}, the support of the sequence $\seqGt$ is determined by the values of $\separationGolayFix$ and $\separationGolay[\indexIteration]$ for $\indexIteration\in\{1,2,\mydots,\numberOfIterations\}$ based on the function in \eqref{eq:shift}. Hence, in order to maintain the properties of \acp{CS}, the partitioning cannot be done arbitrarily. While $\separationGolayFix$ shifts the sequence by prepending zeroes to the sequence, the impact of $\separationGolay[\indexIteration]$ is governed by the monomial $\monomial[\indexIteration]$.  For example, $\separationGolay[1]$ causes the last $2^{\numberOfIterations-1}$ elements to be shifted by $\separationGolay[1]$ as the last $2^{\numberOfIterations-1}$ elements of the corresponding sequence for the monomial $\monomial[1]$ in \eqref{eq:shift}  is $1$. For $\separationGolay[2]$, two clusters of size $2^{\numberOfIterations-2}$ are shifted by $\separationGolay[2]$ since the corresponding sequence for $\monomial[2]$ has two clusters of size $2^{\numberOfIterations-2}$. Under the condition given in \eqref{eq:conditionNoOverlapping}, the impact of $\separationGolayFix$ and $\separationGolay[\indexIteration]$ for $\indexIteration\in\{1,2,\mydots,\numberOfIterations\}$ on the support can be equivalently characterized by the distance between the non-zero clusters.  Due to the enumeration of $\monomial[1],\mydots,\monomial[\numberOfIterations]$ in lexicographic order, the support of the sequence $\seqGt$ changes symmetrically as illustrated in \figurename~\ref{fig:partitioning}. For instance, a standard \ac{CS} constructed with \eqref{eq:imagPartReduced} can be partitioned in two clusters and the distance between the clusters can be arbitrarily chosen as $\separationFreq[1]\in\integersNonnegative$ without losing the features of \acp{CS}. Each cluster can also be  divided into two sub-clusters equivalently and the distance between two sub-clusters can be controlled by another parameter $\separationFreq[2]\in\integersNonnegative$. The same partitioning process can be continued till the $\numberOfIterations$th step and  the amount of separation at the $\indexIteration$th step can be analytically expressed as
\begin{align}
\separationFreq[\indexIteration]=\separationGolay[\indexIteration]-\sum_{\indexForCardinality=\indexIteration+1}^{\numberOfIterations}\separationGolay[\indexForCardinality]~,
\label{eq:separationClosed}
\end{align}
for $\indexIteration=1,2,\mydots,\numberOfIterations$, which defines a bijective mapping between $(\separationFreq[1],\separationFreq[2],\mydots,\separationFreq[\numberOfIterations])$ and $(\separationGolay[1],\separationGolay[2],\mydots,\separationGolay[\numberOfIterations])$.

Given  $\remaningLength=\totalLength-2^{\numberOfIterations}$ zero-valued subcarriers, the first question we need to address is how many different  $(\separationGolayFix,\separationGolay[1],\separationGolay[2],\cdots,\separationGolay[\numberOfIterations])$ sequences exist under the condition in \eqref{eq:conditionNoOverlapping}. The same problem can equivalently be stated as the number of different $(\separationFreqFix,\separationFreq[1],\mydots,\separationFreq[\numberOfIterations])$ sequences such that 
\begin{align}
\separationFreqFix+\separationFreq[1]2^{0}+\separationFreq[2]2^{1}+\cdots+\separationFreq[\numberOfIterations]2^{\numberOfIterations-1}\le\remaningLength~,
\label{eq:conditionSeparation}
\end{align}
where $\separationFreqFix=\separationGolayFix\in\integersNonnegative$ and $\separationFreq[\indexIteration]\in\integersNonnegative$ for $\indexIteration=1,2,\mydots,\numberOfIterations$. This problem is directly related to the subsequent questions: 
\begin{enumerate}
	\item How many extra information bits can be transmitted by changing the support of the synthesized \acp{CS} within the $\totalLength$ subcarriers in the frequency domain?	
	\item For a given $\totalLength$, how to choose $\numberOfIterations$ such that the number of synthesized \acp{CS} with partitioning is maximum?
	\item Can the partitioning maintain the coding gain accomplished for the non-zero elements through the function in \eqref{eq:imagPartReduced}?
	\item What are encoding and decoding procedures if the support of the sequence changes based on information bits?
\end{enumerate}

\subsection{Cardinality Analysis}
Let ${\numberOfEnumarationSequences[\numberOfIterations]}$ denotes the number of different partitioned $\numberOfPointsForPSK$-\ac{PSK} \acp{CS} based on Theorem~\ref{th:reduced}, where their lengths are less than or equal to $\totalLength$.  It is well-known that the number of $\numberOfPointsForPSK$-\ac{PSK} \acp{CS} through the function given in \eqref{eq:imagPartReduced} is
\begin{align}
\unitSequences\triangleq \begin{cases}
\numberOfPointsForPSK^{\numberOfIterations+1}\numberOfIterations!/{2}, & \numberOfIterations>0\\
\numberOfPointsForPSK, & \numberOfIterations=0
\end{cases}~,
\end{align}
where $\numberOfPointsForPSK$ is an even positive integer \cite{davis_1999,paterson_2000}. Let $\cardinalityP[\numberOfIterations][\remaningLength]$ and  $\cardinality[\numberOfIterations][\remaningLengthAnother]$  be the number of distinct sequences $(\separationFreqFix,\separationFreq[1],\separationFreq[2],\cdots,\separationFreq[\numberOfIterations])$ under the condition  in \eqref{eq:conditionSeparation} for a given $\remaningLength$ and the number of distinct sequences $(\separationFreq[1],\separationFreq[2],\cdots,\separationFreq[\numberOfIterations])$ such that $\sum_{\indexIteration=1}^\numberOfIterations\separationFreq[\indexIteration]2^{\indexIteration-1}\le\remaningLengthAnother\in\integersNonnegative$, respectively. ${\numberOfEnumarationSequences[\numberOfIterations]}$ can then be obtained as follows:
\begin{theorem}
\label{co:numberOfSequences}
Let $\numberOfIterations,\remaningLengthAnother\in\integersNonnegative$ and $\numberOfPointsForPSK,\totalLength\in\integersPositive$.
	\begin{align}
		{\numberOfEnumarationSequences[\numberOfIterations]}= \cardinalityP[\numberOfIterations][\totalLength-2^{\numberOfIterations}]\unitSequences~,
		\label{eq:finalCardinality}
	\end{align}
where
	\begin{align}
		\cardinalityP[\numberOfIterations][\remaningLength] =\frac{1}{2}\cardinality[\numberOfIterations+1][2\remaningLength+1]~,
		\label{eq:cardinalityPFinal}
	\end{align}	
and
	\begin{align}
		\cardinality[\numberOfIterations][\remaningLengthAnother]=\begin{cases}
			\sum_{\indexForCardinality=0}^{\remaningLengthAnother} \cardinality[\numberOfIterations-1][\floor*{\frac{\indexForCardinality}{2}}],& \numberOfIterations>1  \\
			\remaningLengthAnother+1, & \numberOfIterations=1 
		\end{cases}~.
		\label{eq:theoremSeperationCardinality}
	\end{align}
\end{theorem}
\begin{proof}
The number of zero-valued elements is  $\remaningLength=\totalLength-2^{\numberOfIterations}$. 
Therefore, for a given set of 
$\{\seqPermutationCompShift,\angleexpAll[1],\mydots,\angleexpAll[\indexIteration],\arbitraryPhaseK\}$, there exist $\cardinalityP[\numberOfIterations][\totalLength-2^{\numberOfIterations}]$ distinct sequences for $(\separationFreqFix,\separationFreq[1],\separationFreq[2],\cdots,\separationFreq[\numberOfIterations])$. The parameter
 $\separationFreqFix$ can range  from $0$ to  $\remaningLength$. Hence, by the definitions of $\cardinality[\numberOfIterations][\remaningLengthAnother]$ and  $\cardinalityP[\numberOfIterations][\remaningLength]$,  $\cardinalityP[\numberOfIterations][\remaningLength]$ can be expressed as
\begin{align}
	\cardinalityP[\numberOfIterations][\remaningLength] = \sum_{\indexForCardinality=0}^{\remaningLength}\cardinality[\numberOfIterations][\remaningLength-\indexForCardinality]=\sum_{\indexForCardinality=0}^{\remaningLength}\cardinality[\numberOfIterations][\indexForCardinality]~.
	\label{eq:cardinalityP}
\end{align}	
The condition $\sum_{\indexIteration=1}^\numberOfIterations\separationFreq[\indexIteration]2^{\indexIteration-1}\le\remaningLengthAnother$ can be re-written as
\begin{align}
	\sum_{\indexIteration=2}^{\numberOfIterations}\separationFreq[\indexIteration]2^{\indexIteration-2}\le\frac{\remaningLengthAnother-\separationFreq[1]}{2}~,
	\label{eq:conditionRewritten}
\end{align}
where $\separationFreq[1]$ can be any integer between $0$ and $\remaningLengthAnother$. The cardinality of the sequences $(\separationFreq[2],\cdots,\separationFreq[\numberOfIterations])$ under the condition in \eqref{eq:conditionRewritten} can be expressed as $\cardinality[\numberOfIterations-1][\floor*{{(\remaningLengthAnother-\separationFreq[1])}/{2}}]$. % as the number of parts reduced by 1. 
Therefore,
\begin{align}
	\cardinality[\numberOfIterations][\remaningLengthAnother]=\sum_{\indexForCardinality=0}^{\remaningLengthAnother} \cardinality[\numberOfIterations-1][\floor*{\frac{\remaningLengthAnother-\indexForCardinality}{2}}]~.\label{eq:Arecursive}
\end{align}
For $\numberOfIterations=1$, $\cardinality[1][\remaningLengthAnother]=\remaningLengthAnother+1$ since there is only $\separationFreq[1]$ that can be any integer  between $0$ and $\remaningLengthAnother$. Also, \eqref{eq:cardinalityP} can be re-expressed as
\begin{align}
	\sum_{\indexForCardinality=0}^{\remaningLength}\cardinality[\numberOfIterations][\indexForCardinality]
	=\frac{1}{2}\sum_{\indexForCardinality=0}^{2\remaningLength+1}\cardinality[\numberOfIterations][\floor*{{\indexForCardinality}/{2}}]=\frac{1}{2}\cardinality[\numberOfIterations+1][2\remaningLength+1]~. \nonumber
\end{align}
\end{proof}

\begin{figure*}
	\centering
	\subfloat[The total number of bits encoded with the non-zero elements and the supports of \acp{CS} ($\numberOfBitsOnTotal$).]{\includegraphics[width =3.5in]{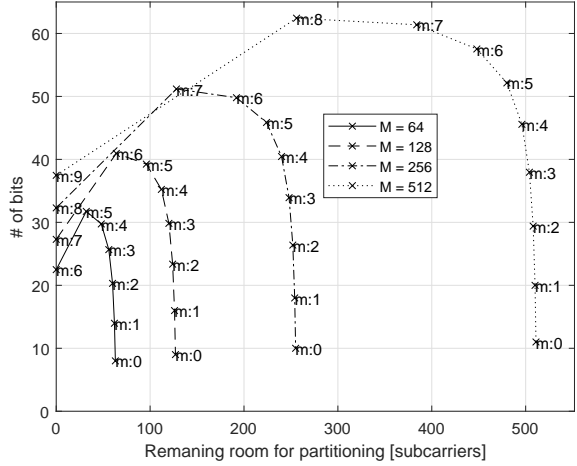}
		\label{subfig:bitsOnPartAndNZ}}
	\subfloat[The total number of bits encoded with the supports of \acp{CS} ($\numberOfBitsOnPartitions$).]{\includegraphics[width =3.5in]{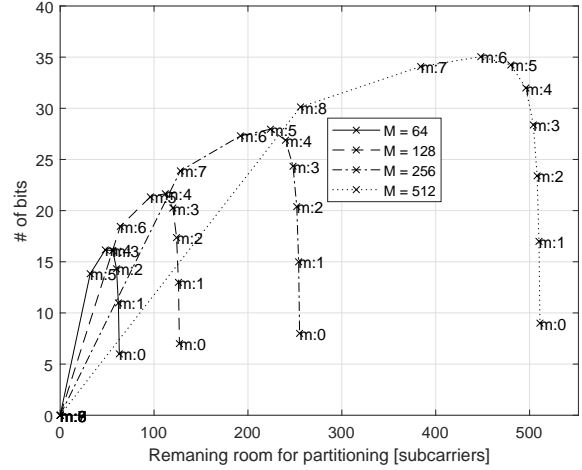}
		\label{subfig:bitsOnPart}}
	\\
	\subfloat[The total number of bits encoded with the non-zero elements of \acp{CS} ($\numberOfBitsOnNZ$).]{\includegraphics[width =3.5in]{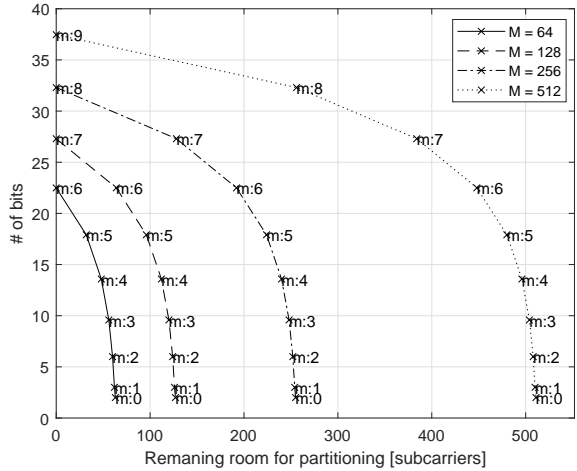}
		\label{subfig:bitsOnNZ}}		
	\subfloat[Maximum achievable SE for a given number of subcarriers.]{\includegraphics[width =3.5in]{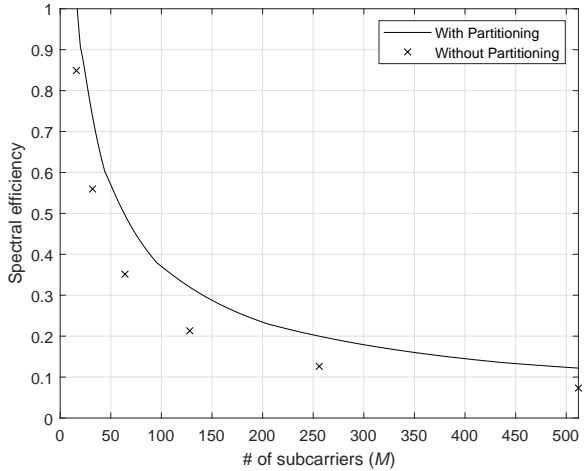}
		\label{subfig:rate}}
	\caption{Partitioning \acp{CS} can increase the number of information bits that can be encoded while increasing the \ac{SE} as compared to the case without partitioning ($\numberOfPointsForPSK=4$).}
	\label{fig:cardinality}
\end{figure*}
Let $\numberOfBitsOnTotal$, $\numberOfBitsOnPartitions$, $\numberOfBitsOnNZ$, and $\spectralEfficient$ denote the number of maximum information bits represented by a \ac{CS}, the number of maximum information bits  conveyed by the support of a \ac{CS}, the number of maximum information bits  carried by the non-zero elements of a \ac{CS}, and the \ac{SE}, respectively. Based on Theorem~\ref{co:numberOfSequences}, $\numberOfBitsOnTotal$, $\numberOfBitsOnPartitions$, $\numberOfBitsOnNZ$, and $\spectralEfficient$ can be calculated as $\numberOfBitsOnTotal=\log_{2}{{\numberOfEnumarationSequences[\numberOfIterations]}}$, $\numberOfBitsOnPartitions=\log_{2}{\cardinalityP[\numberOfIterations][\totalLength-2^{\numberOfIterations}]}$, $\numberOfBitsOnNZ=\log_{2}{\unitSequences}$, and $\spectralEfficient=\floor{\log_{2}{{\numberOfEnumarationSequences[\numberOfIterations]}}}/\totalLength$, respectively. 

In \figurename~\ref{fig:cardinality}, we analyze how $\numberOfBitsOnTotal$, $\numberOfBitsOnPartitions$, and $\numberOfBitsOnNZ$ change for given $\totalLength\in\{64, 128, 256, 512\}$  and $\numberOfIterations\le9$ and provide the maximum achievable \ac{SE}  for a given number of subcarriers $\totalLength\in\{16, 17, \mydots, 512\}$ for $\numberOfPointsForPSK=4$. As shown in \figurename~\ref{fig:cardinality}\subref{subfig:bitsOnPartAndNZ}, the surprising result is that $\numberOfBitsOnTotal$  is not maximum when all available subcarriers are utilized based on the Davis and Jedwab's encoder in \eqref{eq:imagPartReduced}. For example, for $\totalLength=512$, $\numberOfBitsOnTotal$ reaches to $62.43$ bits for $\numberOfIterations=8$ and $\remaningLength=256$ while it is $37.47$ bits for $\numberOfIterations=9$ and $\remaningLength=0$. In other words, the number of different  \acp{CS} by using only \eqref{eq:imagPartReduced} is improved approximately by a factor of $2^{25}$  without changing the alphabet of the non-zero elements of the \ac{CS} (i.e., $\numberOfPointsForPSK$-\ac{PSK}) when the partitioning is taken into account. As shown in \figurename~\ref{fig:cardinality}\subref{subfig:bitsOnPart},  $\numberOfBitsOnPartitions$ first increases by decreasing $\numberOfIterations$ as the remaining room for the partitioning increases. 
However, it sharply decreases after a certain value of $\numberOfIterations$ since the number of clusters (i.e., $2^\numberOfIterations$) decreases for a smaller $\numberOfIterations$. \figurename~\ref{fig:cardinality}\subref{subfig:bitsOnNZ} illustrates how fast $\numberOfBitsOnNZ$ decreases with a smaller $\numberOfIterations$. Since $\numberOfBitsOnPartitions$ increases faster than the degradation of $\numberOfBitsOnNZ$ when $\numberOfIterations$ decreases (up to a certain value of $\numberOfIterations$), $\numberOfBitsOnTotal$ improves. This implies that a large number of information bits can be encoded with the partitioned \acp{CS}. In \figurename~\ref{fig:cardinality}\subref{subfig:rate}, we provide the maximum achievable \ac{SE} with the partitioned \acp{CS} for a given number of subcarriers with a computer search. We observe that the maximum \ac{SE} increases noticeably with the partitioning for a large $\lengthData$ as compared to the cases without partitioning. The partitioned \acp{CS} also exploit the available subcarriers better than the cases without partitioning as their lengths can be non-power-of-two. % The \ac{SE} reaches its maximum for a non-zero $\remaningLength$. For example, when $\numberOfIterations$ is $5$, $6$, $7$, and $8$, the \ac{SE} reaches its maximum, i.e., $0.62$, $0.41$, $0.26$, and $0.16$ bits/sec/Hz, for $\remaningLength$ is $5$, $11$, $21$, and $40$ respectively. We also observe that enabling partitioning allows the transmitter to carry more information bits with \acp{CS}, i.e., $\spectralEfficient(2^\numberOfIterations+\remaningLength)$ bits, since $\remaningLength$ can arbitrarily be chosen.

\subsection{Minimum Distance Analysis}
Let $\minimumDistance$ be the minimum Euclidean distance for the set of the partitioned \acp{CS} for given $\numberOfIterations$ and $\lengthData$, where the non-zero elements follow Davis and Jedwab's construction. We can obtain $\minimumDistance$ as
\begin{align}
	\minimumDistance=\min\{\minimumDistanceNZ,\minimumDistancePart\}~,
\end{align}
where $\minimumDistanceNZ$ and  $\minimumDistancePart$  are the minimum Euclidean distances between two different \acp{CS} such that they have identical and different supports, respectively.

Assume that $\numberOfPointsForPSK\triangleq2^{\numberOfPointsForPSKexp}$ for $\numberOfPointsForPSKexp\in\integersPositive$, and $\numberOfIterations>0$. In \cite{davis_1999}, it was shown that the minimum Lee distance of the code based on \eqref{eq:imagPartReduced} for $\numberOfPointsForPSKexp=1$ and $\numberOfPointsForPSKexp>1$ are $2^{\numberOfIterations-2}$  and $2^{\numberOfIterations-1}$, respectively. Hence, for a given support, $\minimumDistanceNZ$ can be calculated as
\begin{align}
	\minimumDistanceNZ = \begin{cases}
		\sqrt{\lengthData}, & \numberOfPointsForPSKexp=1\\
		\sqrt{2\lengthData}\sin\left(\frac{\pi}{2^\numberOfPointsForPSKexp}\right), & \numberOfPointsForPSKexp\ge1
	\end{cases}~,
	\label{eq:dminphase}
\end{align}
since the magnitude of each non-zero element is scaled to $\magnitudePSK\triangleq\sqrt{\lengthData/2^{\numberOfIterations}}$  as we utilize $\lengthData\ge2^\numberOfIterations$ subcarriers, where $2^\numberOfIterations$ of them are non-zero  and the minimum distance  for $\numberOfPointsForPSK$-PSK alphabet is $2\magnitudePSK\sin\left(\frac{\pi}{\numberOfPointsForPSK}\right)$. 

%Let $\minimumDistancePart$ denotes the minimum Euclidean distance between the \acp{CS} for different $(\separationFreqFix,\separationFreq[1],\mydots,\separationFreq[\numberOfIterations])$ for a given $\lengthData$. %Hence, to maintain the error-rate performance of the standard \acp{CS}, $\minimumDistancePart$ should be similar to $\minimumDistanceNZ$.  
\begin{lemma}
Without any restriction on  $(\separationFreqFix,\separationFreq[1],\mydots,\separationFreq[\numberOfIterations])$,  $\minimumDistancePart$ is bounded as 
\begin{align}
	\minimumDistancePart\ge\sqrt{\frac{\lengthData}{2^{\numberOfIterations-1}}}~.
	\label{eq:dminpart}
\end{align}
\end{lemma}
\begin{proof}
	Let $\seqGt_1$ and $\seqGt_2$ be two \ac{CS} with different supports, where $|\operatorSupport[{\seqGt_1}]|=|\operatorSupport[{\seqGt_2}]|=2^\numberOfIterations$. As $\seqGt_1$ and $\seqGt_2$ have different supports, $|\operatorSupport[{\seqGt_1}]\cap\operatorSupport[{\seqGt_2}]|\le2^\numberOfIterations-2$ must hold. The Euclidean distance between $\seqGt_1$ and $\seqGt_2$ is minimum if the elements on the support are identical, leading to $\minimumDistancePart\ge\magnitudePSK\sqrt{2}$. 
\end{proof}

It can also be shown that  \eqref{eq:dminpart} is tight for same values of $\numberOfIterations$:
\begin{example}
\rm For $\numberOfIterations=3$, consider the cases where $\{(\separationFreqFix,\separationFreq[1],\separationFreq[2],\separationFreq[3])=(0,0,0,0), \seqPermutationCompShift=(3,2,1)\}$ and  $\{(\separationFreqFix,\separationFreq[1],\separationFreq[2],\separationFreq[3])=(0,1,0,0),\seqPermutationCompShift=(2,3,1)\}$ for $\angleexpAll[\indexIteration]=\arbitraryPhaseK=0$ for $\indexIteration=1,2,3$, $\numberOfPointsForPSK=4$, and $\lengthData=9$. The corresponding sequences for the former and the latter cases can be obtained as
\begin{align}
	(1,1,1,-1,1,1,-1,1,0)~,\nonumber\\
	(1,1,1,-1,0,1,-1,1,1)~\nonumber,
\end{align}
respectively. The elements on the fifth and the ninth positions are different while the rest of the elements for any other position are identical. For this example, $\minimumDistanceNZ=3$. Hence, $\minimumDistance=1.5$ since	$\minimumDistancePart=1.5$ and  \eqref{eq:dminpart} is tight.
\end{example}

%\subsubsection{Partitioning  under a minimum distance constraint and cardinality}

For a large $\numberOfIterations$, $\minimumDistancePart$ can be much smaller than $\minimumDistanceNZ$. %However, if , a larger $\minimumDistancePart$ can be achieved. 
To obtain a larger $\minimumDistancePart$, %we propose the following strategy: 
%Assume that $\remaningLength$ is an integer multiple of $2$.
% for $\spacing\in\integersNonnegative$. Let $\referencePoint[\indexReference]=2^{\numberOfIterations-1}+k2^\spacing-0.5$  be the $\indexReference$th reference point  for $k=0,1,2,\mydots,\remaningLength/2^\spacing$. 
we restrict the support of an admissible partitioned \ac{CS} % to be symmetric with respect to the center of resource allocation . %Since two adjacent reference points are separated apart by $2^\spacing$, the minimum distance between two \acp{CS} that do not share the same reference point must be $\magnitudePSK\sqrt{2^{\spacing+1}}$. To retain this distance, we should also ensure that the minimum distance between two \acp{CS} with different supports that share the {\em same} reference point  is at least  $\magnitudePSK\sqrt{2^{\spacing+1}}$. To this end, we exploit the symmetricity of a partitioned \ac{CS} with respect to a reference point. 
and exploit the fact that both halves of a partitioned sequence \ac{CS} have identical partitioning as can be seen in \figurename~\ref{fig:partitioning}. If the minimum distance  for one of the halves is $\magnitudePSK\sqrt{2^{\spacing}}$  for $\spacing\in\integersNonnegative$, the minimum distance for the complete sequence increases to $\magnitudePSK\sqrt{2^{\spacing+1}}$. The problem of obtaining the partitions with the minimum distance of $\magnitudePSK\sqrt{2^{\spacing}}$ for one of the halves is equal to the original problem in \figurename~\ref{fig:partitioning}, where the sequence length is $2^{\numberOfIterations-1}$ and the room for partitioning is $\floor{\remaningLength/2}$. Thus, the number of different supports where $\minimumDistancePart\ge\magnitudePSK\sqrt{2^{\spacing+1}}$ can be calculated recursively as 
\begin{align}
	\cardinalityDistance[\numberOfIterations,\spacing][\remaningLength]= \cardinalityDistance[\numberOfIterations-1,\spacing-1][{\floor{ \remaningLength/2  }}]~,
	\label{eq:cardinalityUnderDistanceConstraint}
\end{align} 
where $\cardinalityDistance[\numberOfIterations,0][\remaningLength]=\cardinalityP[\numberOfIterations][\remaningLength]$.  Equation \eqref{eq:cardinalityUnderDistanceConstraint} leads to the following conclusion:
\begin{corollary}
	Let $\numberOfIterations,\numberOfPointsForPSKexp\in\integersPositive$ and $\totalLength\ge2^\numberOfIterations$. For $\numberOfPointsForPSK=2^\numberOfPointsForPSKexp$ and $0\le\spacing\le\numberOfIterations-1$, ${\numberOfEnumarationSequences[\numberOfIterations,\spacing]}=\cardinalityDistance[\numberOfIterations,\spacing][\totalLength-2^{\numberOfIterations}]\unitSequences$, where $\minimumDistance=\min\{\minimumDistanceNZ,\minimumDistancePart\}$ and $\minimumDistancePart\ge\sqrt{\lengthData/2^{\numberOfIterations-\spacing-1}}$.	%$\minimumDistance=\min\{\minimumDistanceNZ, \sqrt{\lengthData/2^{\numberOfIterations-\spacing-1}}  \} $.
		\label{co:numberOfSequencesDistance}
\end{corollary}
An illustrative example for $\spacing=2$ is provided in \figurename~\ref{fig:partitioningDistance}. Since the support of the final \ac{CS} is restricted to be extended symmetrically in a recursive manner, the minimum distance increases with $\spacing$.

\begin{figure}[t]
	\centering
	{\includegraphics[width =3.3in]{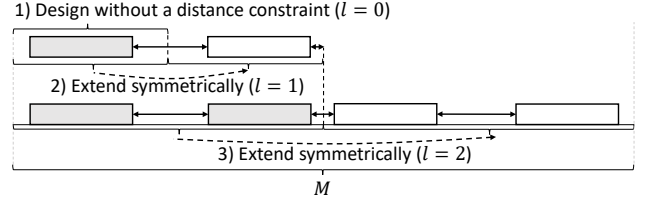}
	}
	\caption{An example of partitioning with a minimum distance constraint ($\spacing=2$). If the minimum distance  for one of the halves is $\magnitudePSK\sqrt{2^{\spacing}}$, the minimum distance for the complete sequence increases to $\magnitudePSK\sqrt{2^{\spacing+1}}$ due to the symmetric support.}
	\label{fig:partitioningDistance}
\end{figure}

\subsection{Bijective Mappings}
To develop an encoder and a decoder based on the partitioned \acp{CS}, the information bits need to be mapped to the sequence $(\separationFreqFix,\separationFreq[1],\mydots,\separationFreq[\numberOfIterations])$ or vice versa.  Since these mappings are not trivial, we develop the algorithms that map a natural number $\integerToBeMapped$ to an admissible $(\separationFreqFix,\separationFreq[1],\mydots,\separationFreq[\numberOfIterations])$  or vice versa  for given $\remaningLength$ and $\spacing$, where $\integerToBeMapped-1$ represents the binary number constructed with a set of information bits in decimal. We need the following definitions:
\begin{definition}\rm
	$\algorithmEncoder[\integerToBeMapped,\remaningLengthAnother,\numberOfIterations]$ is the function that returns  the $\integerToBeMapped$th sequence  $(\separationFreq[1],\mydots,\separationFreq[\numberOfIterations])$ for a given $\remaningLengthAnother$  such that $\separationFreq[1]2^{0}+\separationFreq[2]2^{1}+\cdots+\separationFreq[\numberOfIterations]2^{\numberOfIterations-1}\le\remaningLengthAnother$ and  $\integerToBeMapped\in\{1,2,\mydots,\cardinality[\numberOfIterations][\remaningLengthAnother]\}$.
\end{definition}
\begin{definition}\rm
	 $\algorithmEncoderP[\integerToBeMapped, \remaningLength,\numberOfIterations]$ is the function that returns the $\integerToBeMapped$th sequence  $(\separationFreqFix,\separationFreq[1],\mydots,\separationFreq[\numberOfIterations])$ for a given $\remaningLength$ such that $\separationFreqFix+\separationFreq[1]2^{0}+\separationFreq[2]2^{1}+\cdots+\separationFreq[\numberOfIterations]2^{\numberOfIterations-1}\le\remaningLength$ and   $\integerToBeMapped\in\{1,2,\mydots,\cardinalityP[\numberOfIterations][\remaningLength]\}$.
\end{definition}
\begin{definition}\rm
	$\algorithmEncoderDeltaSeparation[\integerToBeMapped,\remaningLength,\numberOfIterations,\spacing]$ is the function that returns the $\integerToBeMapped$th sequence  $(\separationFreqFix,\separationFreq[1],\mydots,\separationFreq[\numberOfIterations])$ for a given $\remaningLength$ such that $\separationFreqFix+\separationFreq[1]2^{0}+\separationFreq[2]2^{1}+\cdots+\separationFreq[\numberOfIterations]2^{\numberOfIterations-1}\le\remaningLength$ and $\minimumDistancePart\ge\sqrt{\lengthData/2^{\numberOfIterations-\spacing-1}}$ and   $\integerToBeMapped\in\{1,2,\mydots,\cardinalityDistance[\numberOfIterations,\spacing][\remaningLength]\}$.
\end{definition}
\begin{definition}\rm
	The inverse functions of $\algorithmEncoder[\integerToBeMapped,\remaningLengthAnother,\numberOfIterations]$, $\algorithmEncoderP[\integerToBeMapped, \remaningLength,\numberOfIterations]$, and $\algorithmEncoderDeltaSeparation[\integerToBeMapped,\remaningLength,\numberOfIterations,\spacing]$ are $\algorithmDecoder[{(\separationFreq[1],\mydots,\separationFreq[\numberOfIterations])},\remaningLengthAnother,\numberOfIterations]$, $\algorithmDecoderP[{(\separationFreqFix,\separationFreq[1],\mydots,\separationFreq[\numberOfIterations])},\remaningLength,\numberOfIterations]$, and $\algorithmDecoderDeltaSeparation[{(\separationFreqFix,\separationFreq[1],\mydots,\separationFreq[\numberOfIterations])},\remaningLength,\numberOfIterations,\spacing]$, respectively.
\end{definition}

\def\newSumn{n_{\rm test}}
\def\newSum{n_{\rm test}}
\label{subsubsec:bijectionToSeparationUnderDistance}
\begin{algorithm}[t]
	\scriptsize
	\caption{\small Mapping between natural numbers and separations}\label{alg:enum}
	\SetKwInput{KwInput}{Input}                % Set the Input
	\SetKwInput{KwOutput}{Output}              % set the Output
	\DontPrintSemicolon
	
	% Set Function Names
	\SetKwFunction{FencThree}{$\algorithmEncoderDeltaSeparation[\integerToBeMapped,\remaningLength,\numberOfIterations,\spacing]$}
	\SetKwFunction{FencTwo}{$\algorithmEncoderP[\integerToBeMapped, \remaningLength,\numberOfIterations]$}
	\SetKwFunction{FencOne}{$\algorithmEncoder[\integerToBeMapped,\remaningLengthAnother,\numberOfIterations]$}	
	\SetKwFunction{FdecThree}{$\algorithmDecoderDeltaSeparation[{(\separationFreqFix,\separationFreq[1],\mydots,\separationFreq[\numberOfIterations])},\remaningLength,\numberOfIterations,\spacing]$}
	\SetKwFunction{FdecTwo}{$\algorithmDecoderP[{(\separationFreqFix,\separationFreq[1],\mydots,\separationFreq[\numberOfIterations])},\remaningLength,\numberOfIterations]$}
	\SetKwFunction{FdecOne}{$\algorithmDecoder[{(\separationFreq[1],\mydots,\separationFreq[\numberOfIterations])},\remaningLengthAnother,\numberOfIterations]$}

	\SetKwProg{Fn}{Function}{}{}
	\Fn{$(\separationFreqFix,\separationFreq[1],\mydots,\separationFreq[\numberOfIterations])=\FencThree$}{
		
		\eIf{$\spacing=0$}
		{
			$(\separationFreqFix,\separationFreq[1],\mydots,\separationFreq[\numberOfIterations])=\FencTwo$
		}{
%		Calculate $\NsumAll=\sum_{\indexReference=0}^{\chosenReferencePoint} \cardinalityDistance[\numberOfIterations-1,\spacing-1][{\roomForAReferencePoint[\indexReference]}]\le\integerToBeMapped$ for the largest $\chosenReferencePoint$\;
		$(\separationFreq[1],\mydots,\separationFreq[\numberOfIterations])\leftarrow\algorithmEncoderDeltaSeparation[\integerToBeMapped,{\floor{\remaningLength/2}},\numberOfIterations-1,\spacing-1]$\;
		$\separationFreq[1]\leftarrow2\separationFreq[1]$\;
		$\separationFreqFix\leftarrow\floor{\remaningLength/2}-(\sum_{i=0}^{\numberOfIterations-1}\separationFreq[i+1]2^i)/2$
		}
	}\;

	\SetKwProg{Fn}{Function}{}{}
	\Fn{$(\separationFreqFix,\separationFreq[1],\mydots,\separationFreq[\numberOfIterations])=\FencTwo$}{
%		$\separationFreqFix\leftarrow0$, $\NsumAll \leftarrow0$,			$\newSumn\leftarrow\cardinality[\numberOfIterations][\remaningLength-1]$\;
		Calculate $\NsumAll=\sum_{\indexForCardinality=0}^{\chosenReferencePointOther}\cardinality[\numberOfIterations][\remaningLength-\indexForCardinality]\le\integerToBeMapped$ for the largest $\chosenReferencePointOther$\;
%		\While{$\integerToBeMapped>\NsumAll+\newSumn$}{
%			$\separationFreqFix\leftarrow\separationFreqFix+1$\;
%			$\NsumAll\leftarrow\NsumAll+\newSumn$\;
%			$\newSumn\leftarrow\cardinality[\numberOfIterations][\remaningLength-\separationFreqFix]$\;
%		}
		$\separationFreqFix\leftarrow\chosenReferencePointOther$\;
		$(\separationFreq[1],\mydots,\separationFreq[\numberOfIterations])\leftarrow\algorithmEncoder[\integerToBeMapped-\NsumAll,\remaningLength-\chosenReferencePointOther,\numberOfIterations]$
	}\;

	\SetKwProg{Fn}{Function}{}{}
	\Fn{$(\separationFreq[1],\mydots,\separationFreq[\numberOfIterations])=\FencOne$}{
		\eIf{$\numberOfIterations=1$}
		{
			$\separationFreq[1]\leftarrow\integerToBeMapped-1$\;
		}{
			%$\separationFreq[1]\leftarrow0$, $\NsumAll \leftarrow0$,			$\newSum\leftarrow\cardinality[\numberOfIterations-1][\floor{\remaningLengthAnother/2}]$\;
			Calculate $\NsumAll=\sum_{\indexForCardinality=0}^{\chosenReferencePointOther}\cardinality[\numberOfIterations][\floor{(\remaningLengthAnother-\indexForCardinality)/2}]$ for the largest $\chosenReferencePointOther$\;
%			\While{$\integerToBeMapped>\NsumAll+\newSum$}{
%				$\separationFreq[1]\leftarrow\separationFreq[1]+1$\;
%				$\NsumAll\leftarrow\NsumAll+\newSum$\;
%				$\newSum\leftarrow\cardinality[\numberOfIterations-1][\floor{(\remaningLengthAnother-\separationFreq[1])/2}]$\;
%			}
			$\separationFreq[1]\leftarrow\chosenReferencePointOther$\;
			$(\separationFreq[2],\mydots,\separationFreq[\numberOfIterations])=\algorithmEncoder[\integerToBeMapped-\NsumAll,\floor{(\remaningLengthAnother-\chosenReferencePointOther)/2},\numberOfIterations-1]$
		}
	}\;

	\SetKwProg{Fn}{Function}{}{}
	\Fn{$\integerToBeMapped=\FdecThree$}{
		\eIf{$\spacing=0$}
		{
			$\integerToBeMapped=\FdecTwo$
		}{
			%$\chosenReferencePoint\leftarrow(\separationFreqFix+(\sum_{i=0}^{\numberOfIterations-1}\separationFreq[i+1]2^i)/2)/2^\spacing$\;
			$\separationFreq[1]\leftarrow\separationFreq[1]/2$\;
			%Calculate $\NsumAll=\sum_{\indexReference=0}^{\chosenReferencePoint-1}\cardinalityDistance[\numberOfIterations-1,\spacing-1][{\roomForAReferencePoint[\indexReference]}]$\;
			$\integerToBeMapped\leftarrow\algorithmDecoderDeltaSeparation[{(\separationFreq[1],\mydots,\separationFreq[\numberOfIterations])},{{\floor{\remaningLength/2}}},\numberOfIterations-1,\spacing-1]$
		}
	}\;
	
	\SetKwProg{Fn}{Function}{}{}
	\Fn{$\integerToBeMapped=\FdecTwo$}{
		Calculate $\NsumAllGiven[\separationFreqFix]=\sum_{\indexForCardinality=0}^{\separationFreqFix-1}\cardinality[\numberOfIterations][\remaningLength-\indexForCardinality]$\;
		$\integerToBeMapped=\NsumAllGiven[\separationFreqFix]+\algorithmDecoder[{(\separationFreq[1],\mydots,\separationFreq[\numberOfIterations])},\remaningLength-\separationFreqFix,\numberOfIterations]$
	}\;
	
	\SetKwProg{Fn}{Function}{}{}
	\Fn{$\integerToBeMapped=\FdecOne$}{
		$\integerToBeMapped\leftarrow1$\;
		\eIf{$\numberOfIterations=1$}
		{
			$\integerToBeMapped\leftarrow\separationFreq[1]+1$\;
		}{
			Calculate $\NsumAllGiven[{\separationFreq[1]}]=\sum_{\indexForCardinality=0}^{\separationFreq[1]-1}\cardinality[\numberOfIterations-1][\floor{(\remaningLengthAnother-\indexForCardinality)/2}]$\;
			$\integerToBeMapped\leftarrow\NsumAllGiven[{\separationFreq[1]}]+\algorithmDecoder[{(\separationFreq[2],\mydots,\separationFreq[\numberOfIterations])},\floor{(\remaningLengthAnother-\separationFreq[1])/2},\numberOfIterations-1]$
		}
	}\;
\end{algorithm}

\subsubsection{Mapping from natural numbers to separations}
\label{subsubsec:bijectionToSeparation}

For $\spacing>0$, the function $\algorithmEncoderDeltaSeparation[\integerToBeMapped,\remaningLength,\numberOfIterations,\spacing]$ first  obtains the sequence $(\separationFreq[1],\mydots,\separationFreq[\numberOfIterations])$ by calling itself as $\algorithmEncoderDeltaSeparation[{\integerToBeMapped},{\floor{\remaningLength/2}},\numberOfIterations-1,\spacing-1]$. After setting $\separationFreq[1]\leftarrow2\separationFreq[1]$ (due to the symmetric support of the halves), it calculates $\separationFreqFix$ as $\separationFreqFix\leftarrow\floor{\remaningLength/2}-(\sum_{i=0}^{\numberOfIterations-1}\separationFreq[i+1]2^i)/2$. For $\spacing=0$, the corresponding sequence is calculated with $\algorithmEncoderP[\integerToBeMapped,\remaningLength,\numberOfIterations]$ as there is no restriction on the support.

The function $\algorithmEncoderP[\integerToBeMapped, \remaningLength,\numberOfIterations]$ utilizes \eqref{eq:cardinalityP}. 
It calculates  $\separationFreqFix$ as the largest $\chosenReferencePointOther$ such that $\NsumAll=\sum_{\indexForCardinality=0}^{\chosenReferencePointOther}\cardinality[\numberOfIterations][\remaningLength-\indexForCardinality]\le\integerToBeMapped$. Hence, the remaining sequence $(\separationFreq[1],\mydots,\separationFreq[\numberOfIterations])$ can be identified as $(\separationFreq[1],\mydots,\separationFreq[\numberOfIterations])=\algorithmEncoder[{\integerToBeMapped-\NsumAll},\remaningLength-\chosenReferencePointOther,\numberOfIterations]$.

The function $\algorithmEncoder[{\integerToBeMapped},\remaningLengthAnother,\numberOfIterations]$ calculates  $\separationFreq[1]$ as the largest $\chosenReferencePointOther$ such that $\NsumAll=\sum_{\indexForCardinality=0}^{\chosenReferencePointOther}\cardinality[\numberOfIterations][\floor{(\remaningLengthAnother-\indexForCardinality)/2}]\le\integerToBeMapped$ for $\numberOfIterations>1$ based on  \eqref{eq:Arecursive}. It obtains the remaining sequence $(\separationFreq[2],\mydots,\separationFreq[\numberOfIterations])$  from $\algorithmEncoder[{\integerToBeMapped-\NsumAll},\floor{(\remaningLengthAnother-\chosenReferencePointOther)/2},\numberOfIterations-1]$. For $\numberOfIterations=1$, $\separationFreq[1]$ is $\integerToBeMapped-1$.

The pseudocodes for $\algorithmEncoder[\integerToBeMapped,\remaningLengthAnother,\numberOfIterations]$, $\algorithmEncoderP[\integerToBeMapped, \remaningLength,\numberOfIterations]$, and $\algorithmEncoderDeltaSeparation[\integerToBeMapped,\remaningLength,\numberOfIterations,\spacing]$ are provided in Algorithm~\ref{alg:enum}.

\subsubsection{Mapping from separations to natural numbers}
The function $\algorithmDecoderDeltaSeparation[{(\separationFreqFix,\separationFreq[1],\mydots,\separationFreq[\numberOfIterations])},\remaningLength,\numberOfIterations,\spacing]$
% first calculates $\chosenReferencePoint=(\separationFreqFix+(\sum_{i=0}^{\numberOfIterations-1}\separationFreq[i+1]2^i)/2)/2^\spacing$%. After calculating $\NsumAll=\sum_{\indexReference=0}^{\chosenReferencePoint-1} \cardinalityDistance[\numberOfIterations-1,\spacing-1][{\roomForAReferencePoint[\indexReference]}]$, it then 
 returns  $\integerToBeMapped$ as $\algorithmDecoderDeltaSeparation[{(\separationFreq[1]/2,\mydots,\separationFreq[\numberOfIterations])},{\floor{\remaningLength/2}},\numberOfIterations-1,\spacing-1]$. For $\spacing=0$, $\integerToBeMapped$  is calculated with $\algorithmDecoderP[{(\separationFreqFix,\separationFreq[1],\mydots,\separationFreq[\numberOfIterations])},\remaningLength,\numberOfIterations]$.

\label{subsubsec:bijectionToNumbers}
The function $\algorithmDecoderP[{(\separationFreqFix,\separationFreq[1],\mydots,\separationFreq[\numberOfIterations])},\remaningLength,\numberOfIterations]$ calculates how much $\separationFreqFix$ contributes to $\integerToBeMapped$ as $\NsumAllGiven[\separationFreqFix]=\sum_{\indexForCardinality=0}^{\separationFreqFix-1}\cardinality[\numberOfIterations][\remaningLength-\indexForCardinality]$ by using \eqref{eq:cardinalityP}. It obtains $\integerToBeMapped$ as $\NsumAll+\algorithmDecoder[{(\separationFreq[1],\mydots,\separationFreq[\numberOfIterations])},\remaningLength-\separationFreqFix,\numberOfIterations]$.

The function $\algorithmDecoder[{(\separationFreq[1],\mydots,\separationFreq[\numberOfIterations])},\remaningLengthAnother,\numberOfIterations]$ calculates the contribution of $\separationFreq[1]$  as $\NsumAllGiven[{\separationFreq[1]}]=\sum_{\indexForCardinality=0}^{\separationFreq[1]-1}\cardinality[\numberOfIterations-1][\floor*{{(\remaningLengthAnother-\indexForCardinality)}/{2}}]$ for $\numberOfIterations>1$ by exploiting \eqref{eq:Arecursive}. It obtains $\integerToBeMapped$ as $\integerToBeMapped=\NsumAll+\algorithmDecoder[{(\separationFreq[2],\mydots,\separationFreq[\numberOfIterations])},\floor*{{(\remaningLengthAnother-\separationFreq[1])}/{2}},\numberOfIterations-1]$. For $\numberOfIterations=1$, $\integerToBeMapped$ is $\separationFreq[1]+1$.

The pseudocodes for $\algorithmDecoder[{(\separationFreq[1],\mydots,\separationFreq[\numberOfIterations])},\remaningLengthAnother,\numberOfIterations]$, $\algorithmDecoderP[{(\separationFreqFix,\separationFreq[1],\mydots,\separationFreq[\numberOfIterations])},\remaningLength,\numberOfIterations]$, and $\algorithmDecoderDeltaSeparation[{(\separationFreqFix,\separationFreq[1],\mydots,\separationFreq[\numberOfIterations])},\remaningLength,\numberOfIterations,\spacing]$ are provided in Algorithm~\ref{alg:enum}.

\section{Encoder and Decoder}
\label{sec:encAndDec}
In this section, we use the partitioned \acp{CS} to synthesize low-\ac{PAPR} \ac{OFDM} symbols and discuss the corresponding encoding and decoding operations.

\subsection{Encoder}
We first split the information bits into two bit sequences that are mapped to $(\arbitraryPhaseK,\angleexpAll[1],\mydots,\angleexpAll[\numberOfIterations])$ and  $\{\seqPermutationCompShift, (\separationFreqFix,\separationFreq[1],\mydots,\separationFreq[\numberOfIterations])\}$, i.e.,  $\bitsForPhase$ and $\bitsForSep$, respectively.
%, where theirs lengths are $\floor{\log_2(\numberOfIterations!/2)}$, $\floor{\log_2(	\cardinalityP[\numberOfIterations][\totalLength-2^{\numberOfIterations}])}$, and $(\numberOfIterations+1)\log_{2}\numberOfPointsForPSK$, respectively.
We then convert $\bitsForSep$ to a decimal number $\decimalForSepAndPerm$. % and map the corresponding decimals ///to $\seqPermutationCompShift$ and  $(\separationFreqFix,\separationFreq[1],\mydots,\separationFreq[\numberOfIterations])$. 
To obtain $\seqPermutationCompShift$ and $(\separationFreqFix,\separationFreq[1],\mydots,\separationFreq[\numberOfIterations])$, we decompose  $\decimalForSepAndPerm = \decimalForSep\numberOfIterations!/2+\decimalForPerm$ such that $\decimalForPerm<\numberOfIterations!/2$ for $\decimalForPerm,\decimalForSep\in\integersNonnegative$. We then utilize factoradic based on Lehmer code \cite{Lehmer1960TeachingCT} to obtain $\seqPermutationCompShift$ from $\decimalForPerm$. To avoid using $\seqPermutationCompShift$ and its reversed version, we assume that $\permutationMono[1]>\permutationMono[\numberOfIterations]$. We obtain $(\separationFreqFix,\separationFreq[1],\mydots,\separationFreq[\numberOfIterations])$ with $\algorithmEncoderDeltaSeparation[\decimalForSep+1,\remaningLength,\numberOfIterations,\spacing]$  and control the minimum Euclidian distance for the partitioned \acp{CS} with $\spacing$. For $(\arbitraryPhaseK,\angleexpAll[1],\mydots,\angleexpAll[\numberOfIterations])$, the bit mapping is done based on a Gray mapping, e.g.,  $00\rightarrow0$, $01\rightarrow1$, $10\rightarrow3$, and $11\rightarrow2$ for $\numberOfPointsForPSK=4$. After obtaining $\seqPermutationCompShift$, $(\arbitraryPhaseK,\angleexpAll[1],\mydots,\angleexpAll[\numberOfIterations])$, and  $(\separationFreqFix,\separationFreq[1],\mydots,\separationFreq[\numberOfIterations])$, the corresponding baseband \ac{OFDM} symbol   is calculated as $\OFDMinTime[\seqGt][\timeVar]$ where $\timeVar$ is a partitioned \ac{CS} that encodes the information bits. The transmitter diagram is provided in \figurename~\ref{fig:txrx}\subref{subfig:tx}. The transmitter maps  the elements of a standard CS to the subcarriers chosen based on $\bitsForSep$.

%\begin{example}
%	content...
%\end{example}

\begin{figure*}
	\centering
	\subfloat[{Transmitter diagram.}]{\includegraphics[width =5.7in]{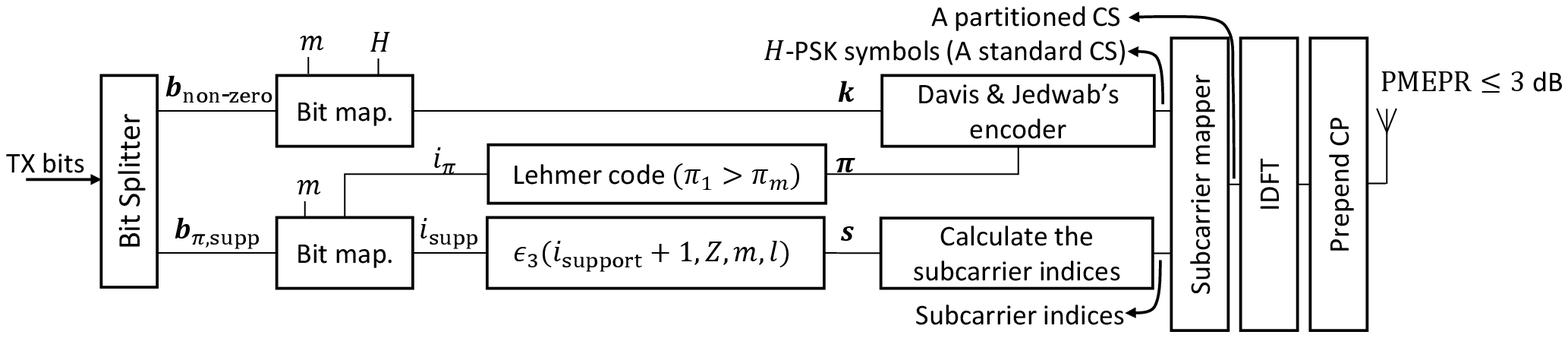}
		\label{subfig:tx}}\\
	\subfloat[{Receiver diagram.}]{\includegraphics[width =5.7in]{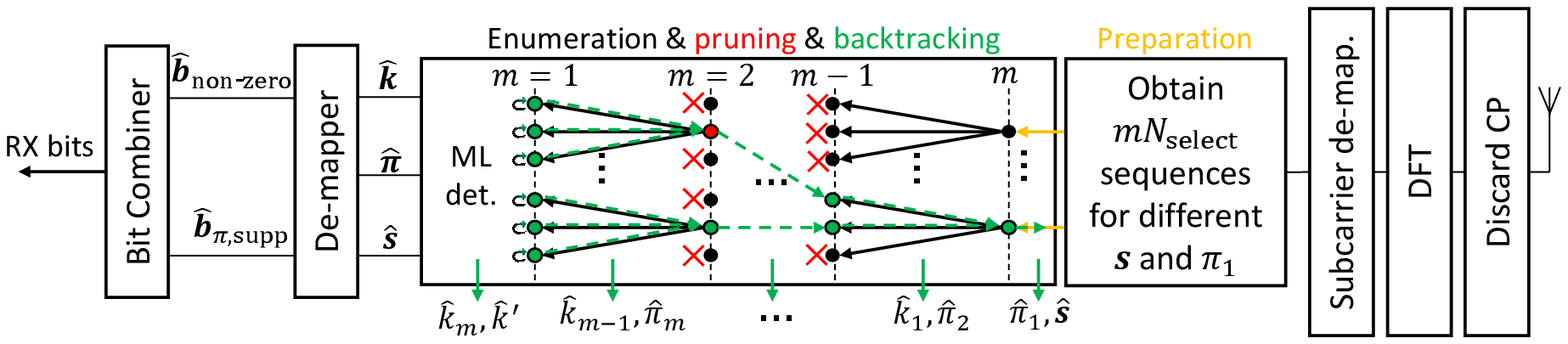}
		\label{subfig:rx}}	
	\caption{{Transmitter and receiver diagrams for partitioned CSs.}}
	\label{fig:txrx}
\end{figure*}

\subsection{Decoder}

%for a given $\seqS=(\separationFreqFix,\separationFreq[1],\mydots,\separationFreq[\numberOfIterations])$,  we first concatenate the elements at the corresponding indices. The $\varMonomial$th element of the concatenated sequence can be hypothesized as $\receivedElement[\varMonomial|\seqS] = \channel[\varMonomial|\seqS]\exponentialBase^{\constantj \funcfForFinalPhase(\seqx)} +\noise[\varMonomial|\seqS] $,

We can express the $\varMonomial$th element of the received signal as $\receivedElement[\varMonomial] = \channel[\varMonomial]\eleGt[\varMonomial] +\noise[\varMonomial]$, where $ \channel[\varMonomial]$ and $\noise[\varMonomial]$ are the complex fading channel and the noise coefficients, respectively.  Assuming that the channel coefficients are available at the receiver, a \ac{ML} decoder corresponds to a minimum distance decoder for \ac{AWGN}, i.e.,
\begin{align}
	\{\parametersImagEstimate,\seqStilde \} &= \arg\min_{\{\parametersImag,\seqS \}} \sum_{\varMonomial=0}^{2^\numberOfIterations-1}   |\channel[{\functionIndex[\seqS][\varMonomial]}]\exponentialBase^{ \constantj \funcfForFinalPhaseDec(\varMonomial; \parametersImag)} -  \receivedElement[{\functionIndex[\seqS][\varMonomial]}] |^2	\label{eq:mlIni}\\
	&= \arg\max_{\{\parametersImag,\seqS \}} 	\Re\left\{ \sum_{\varMonomial=0}^{2^\numberOfIterations-1}  \exponentialBase^{-\constantj \funcfForFinalPhaseDec(\varMonomial; \parametersImag)} \weightedElement[{\functionIndex[\seqS][\varMonomial]}]\right\}-\channelSum[{\seqS}]~,
\label{eq:mlNew}	
\end{align}
where $\weightedElement[{\functionIndex[\seqS][\varMonomial]}] = \channel[{\functionIndex[\seqS][\varMonomial]}]^* \receivedElement[{\functionIndex[\seqS][\varMonomial]}]$, $
\channelSum[{\seqS}]=\sum_{\varMonomial=0}^{2^{\numberOfIterations}\shortMinus 1}   \channelPowerElement[{\functionIndex[\seqS][\varMonomial]}]
$ is the channel quality information, 
$\channelPowerElement[{\functionIndex[\seqS][\varMonomial]}]={|\channel[{\functionIndex[\seqS][\varMonomial]}]|^2}/{2}$, $\parametersImag = \{\seqPermutationCompShift,\arbitraryPhaseK,\angleexpAll[1],\mydots,\angleexpAll[\numberOfIterations] \}$, $\seqS=(\separationFreqFix,\separationFreq[1],\mydots,\separationFreq[\numberOfIterations])$, and $\functionIndex[\seqS][\varMonomial]=\varMonomial+\funcfForFinalPhaseDec(\varMonomial)$.

To solve \eqref{eq:mlNew}, we consider the principle discussed in \cite{sahin_2020gm} and modify it by introducing simplications due to the \ac{PSK} alphabet for the non-zero elements of the partitioned \acp{CS}. As done in \cite{sahin_2020gm}, we first decompose $\funcfForFinalPhase(\seqx;\parametersImag)$ as  $\funcfForFinalPhase(\seqx;\parametersImag)=\funcgForFinalPhaseDerivate[\referenceDerivative][\seqx^\referenceDerivative;\parametersImagg]+\monomial[\referenceDerivative]\funcfForFinalPhaseDerivate[\referenceDerivative][\seqx^\referenceDerivative;\parametersImagf]$ for $\referenceDerivative=\{1,2,\mydots,\numberOfIterations\}$, where $\seqx^\referenceDerivative\triangleq(\monomial[1],\monomial[2], \mydots, \monomial[\referenceDerivative-1],\monomial[\referenceDerivative+1], \mydots, \monomial[\numberOfIterations])$ and $\funcfForFinalPhaseDerivate[\referenceDerivative][\seqx^\referenceDerivative] \triangleq \frac{{\partial \funcfForANF(\seqx)}}{\partial \monomial[{\referenceDerivative}]}$. By using this decomposition, we then re-write the objective function in \eqref{eq:mlNew} as
\begin{align}
	\max_{\{\parametersImagg,\parametersImagf,\seqS \}}
	\Re\left\{ \sum_{\varMonomial=0}^{2^{\numberOfIterations-1}\shortMinus 1}  \exponentialBase^{-\constantj \funcgForFinalPhaseDerivateDec[\referenceDerivative][\seqx;\parametersImagg] } \weightedElementAfterSum[{\varMonomial},\referenceDerivative]\right\}-\channelSum[{\seqS}]~,
	\label{eq:mlNeww}
\end{align}
where 
\begin{align}
	\weightedElementAfterSum[\varMonomial]= \weightedElement[{\functionIndex[\seqS][{\mappingFunction[\varMonomial]}]} ] + \weightedElement[{\functionIndex[\seqS][{\mappingFunction[\varMonomial]}+2^{\numberOfIterations-\referenceDerivative}]}  ]\exponentialBase^{-\constantj \funcfForFinalPhaseDerivateDec[\referenceDerivative][\varMonomial;\parametersImagf] }~,
	\label{eq:newSeq}
\end{align}
and $\mappingFunction[\varMonomial]$ maps the integers between $0$ and $2^{\numberOfIterations-1}-1$ to the values of $ \sum_{\indexFirstOrderMonomial=1}^{\numberOfIterations}\monomial[\indexFirstOrderMonomial]2^{\numberOfIterations-\indexFirstOrderMonomial}$    in ascending order for $\monomial[\referenceDerivative]=0$. For ${\permutationMono[1]}=\referenceDerivative$, $\funcfForFinalPhaseDerivate[\referenceDerivative][\seqx^\referenceDerivative]$ can be calculated as $\angleexpAll[1]+\frac{\numberOfPointsForPSK}{2}\monomial[{\permutationMono[2]}]$, which leads to $\parametersImagf=\{\angleexpAll[1],\permutationMono[2]\}$ for a given $\referenceDerivative$. For a hypothesized $\seqS$,  $\angleexpAll[1]$, $\permutationMono[2]$, and $\referenceDerivative$, finding $\parametersImagg$ that solves \eqref{eq:mlNeww} is
equivalent to the problem \eqref{eq:mlNew}, but  the length of the sequence $(\weightedElementAfterSum[\varMonomial])_{\varMonomial=0}^{2^{\numberOfIterations-1}-1}$ is $2^{\numberOfIterations-1}$. Hence, by repeating the same procedure, a recursive \ac{ML} decoder can be obtained. The main bottleneck of this approach is the exponential growth of the enumerated parameters. To address this issue, we terminate the unpromising branches early as done in \cite{sahin_2020gm}. As compared to \cite{sahin_2020gm}, we run the algorithm for the promising values for the sequence $\seqS$ in a parallel. We simplify the decoder by removing steps that are related to high-order constellations. We also reduce the complexity at the backtracking stage as we consider unimodular non-zero elements.

%\subsubsection{Algorithm}
%\label{subsec:algorithm}
\begin{algorithm}[t]
	\scriptsize
	\caption{\small A recursive decoder for the partitioned \acp{CS}}\label{alg:decoder}
	\SetKwInput{KwInput}{Input}                % Set the Input
	\SetKwInput{KwOutput}{Output}              % set the Output
	\DontPrintSemicolon
	
	% Set Function Names
	\SetKwFunction{FMain}{main}
	\SetKwFunction{Fdecoder}{dec}
	\SetKwFunction{Fesc}{offsets}
	
	\SetKwProg{Fn}{Function}{}{}
	\Fn{$(\bitsForPhase,\bitsForSep)$=\FMain$((\channel[\varMonomial])_{\varMonomial=0}^{\totalLength-1}),(\receivedElement[\varMonomial])_{\varMonomial=0}^{\totalLength-1},\numberOfIterations,\numberOfPointsForPSK,\remaningLength,\spacing$)}{
		Calculate $\weightedElement[\varMonomial] \leftarrow \channel[\varMonomial]^* \receivedElement[\varMonomial]$ and
		$\channelPowerElement[\varMonomial]\leftarrow{|\channel[\varMonomial]|^2}/{2}$\;
		Populate $\setOfweightedSequence$,			$\setOfweightedChannel$, and $\setOfpermutations$ for different $\seqS$ and $\permutationMono[1]$\; 
		Run $(\indexOptimum,\seqPermutationCompShift,\angleexpAll[1],\mydots,\angleexpAll[\numberOfIterations],\arbitraryPhaseK)=$\Fdecoder{$\setOfweightedSequence$, $\setOfweightedChannel$, $\setOfpermutations$}\;
		Obtain $\seqS$ based on $\indexOptimum$\;
		Calculate $\decimalForPerm $ from $\seqPermutationCompShift$ (permutation to integer)\;
		Calculate $\decimalForSep=\algorithmDecoderDeltaSeparation[{\seqS},\remaningLength,\numberOfIterations,\spacing]-1$\;
		Calculate $\decimalForSepAndPerm = \decimalForSep \numberOfIterations!/2 +\decimalForPerm$\;
		Calculate $\bitsForPhase$ from  $(\angleexpAll[1],\mydots,\angleexpAll[\numberOfIterations], \arbitraryPhaseK)$\;
		Calculate $\bitsForSep$ from $\decimalForSepAndPerm$\;
	}\;
	
	\SetKwProg{Fn}{Function}{}{}
	\Fn{  $(\indexOptimum,\seqPermutationCompShiftD,\angleexpAllBitD[1],\mydots,\angleexpAllBitD[\numberOfIterations],\arbitraryPhaseKBitD)=$\Fdecoder{$\setOfweightedSequence$, $\setOfweightedChannel$, $\setOfpermutations$}}{

		Enumerate $\setOfweightedSequenceSub$,$\setOfweightedChannel$, and	$\setOfpermutationsSub$

		\eIf{$\numberOfIterations=1$}{
			%	\For{$\indexSequence=1$ \KwTo $\numberOfSequences$}{
			Set $\indexOptimum$ as  the index of the best sequence in $\setOfweightedSequenceSub$ \;
			Obtain $\arbitraryPhaseKBitD$ and $\angleexpAllBitD[1]$ based on ML detection for $\indexOptimum$th sequence\;
			Set $\seqPermutationCompShiftD=(1)$		%	}
		}{
			Prune  $\setOfweightedSequenceSub$, $\setOfweightedChannelSub$, $\setOfpermutationsSub$\; 
			Populate the indices of the chosen sequences in $\setOfgoodSequenceIndexSub$\;
			Run $(i,\seqPermutationCompShiftSub,\angleexpAllBitD[2],\mydots,\angleexpAllBitD[\numberOfIterations],\arbitraryPhaseKBitD)=$\Fdecoder{$\setOfweightedSequenceSub$, $\setOfweightedChannelSub$, $\setOfpermutationsSub$}\;
			Calculate $\indexOptimumPre$ as the $i$th element of $\setOfgoodSequenceIndexSub$\;
			Calculate $\indexOptimum$ from $\indexOptimumPre$ (based on the enumaration order)\;
			Calculate $\permutationMonoD[{1}]$ as $ \text{the }\indexOptimum\text{th element of }\setOfpermutationsSub$\;
			Set $(\permutationMonoD[2],\mydots,\permutationMonoD[\numberOfIterations])$ as $  (1,\mydots,\permutationMonoD[1]-1,\permutationMonoD[1]+1,\mydots,\numberOfIterations)_{\seqPermutationCompShiftSub}$\;
			Calculate $\angleexpAllBitD[1]$ from $\indexOptimum$  (based on the enumaration order)\;
			
		}
	}\;
\end{algorithm}
The decoder can be described in the following fives stages:
\subsubsection{Preparation}
We first choose the most likely $\numberOfSelectSeqences=\min\{\cardinalityDistance[\numberOfIterations,\spacing][\remaningLength],\numberOfSelectSeqencesMax\}$ separations based on the metric given by
\begin{align}
	\sum_{\varMonomial=0}^{2^{\numberOfIterations}-1} \max_{{c_\varMonomial}\in\integers_\numberOfPointsForPSK}(\Re\{\exponentialBase^{-\constantj c_\varMonomial}\weightedElement[\varMonomial]\})-\weightedChannel[\indexEnumaration]~,
	\label{eq:score}
\end{align}
i.e., a high \ac{SNR} estimation of \eqref{eq:mlNew}, where $\numberOfSelectSeqencesMax$ is the maximum number of separations that are considered in the decoding procedure. Since the receiver does not know $\permutationMono[1]\in\{1,2,\mydots\numberOfIterations\}$ in advance, we repeat the selected sequences $\numberOfIterations$ times and populate them in $\setOfweightedSequence$, which leads to $\numberOfSequences=\numberOfIterations\numberOfSelectSeqences$ sequences. The corresponding channel quality information and the hypothesized $\permutationMono[1]$ are listed in $\setOfweightedChannel$ and $\setOfpermutations$, respectively.

\subsubsection{Enumeration}
The decoder enumerates $(\weightedElementAfterSum[\varMonomial])_{\varMonomial=0}^{2^{\numberOfIterations-1}-1}$ based on \eqref{eq:newSeq} for $\numberOfIterations-1$ and $\numberOfPointsForPSK$ options for $\permutationMono[2]$ and $\angleexpAll[1]$, respectively. It then populates the resulting sequences in $\setOfweightedSequenceSub$. It also keeps the channel quality information and chosen  $\permutationMono[2]$ in $\setOfweightedChannelSub$ and $\setOfpermutationsSub$, respectively.

\subsubsection{Pruning}
To avoid exponential growth, the decoder prunes  $\setOfweightedSequenceSub$ (and the corresponding elements of $\setOfweightedChannelSub$ and $\setOfpermutationsSub$). To this end, it calculates the score given by
\begin{align}
	\sum_{\varMonomial=0}^{2^{\numberOfIterations-1}-1} \max_{{c_\varMonomial}\in\integers_\numberOfPointsForPSK}(\Re\{\exponentialBase^{-\constantj c_\varMonomial}\weightedElement[\varMonomial]'\})-\weightedChannel[\indexEnumaration]~,
	\label{eq:scoreP}
\end{align}
for  $ \weightedElement[\varMonomial]' \in\setOfweightedSequenceSub$.
The detector then chooses $\numberOfGoodSeqences$ sequences based on the score and terminates others. For backtracking, it also lists the indices of the chosen sequences in $\setOfgoodSequenceIndexSub$. It then calls itself. It repeats the enumeration and pruning stages till it reaches to $\numberOfIterations=1$.

\subsubsection{Backtracking}
For $\numberOfIterations=1$, the decoder performs \ac{ML} detection for the sequences in $\setOfweightedSequenceSub$ and finds the $\indexOptimum$th sequence in  $\setOfweightedSequenceSub$ that maximizes the likelihood. It then returns the sequence index $\indexOptimum$  and the detected parameters $\angleexpAllBitD[1]$, $\arbitraryPhaseKBitD$, and $\seqPermutationCompShiftD=(1)$. For $\numberOfIterations\ge2$, the decoder first identifies the chosen sequence by using $\setOfgoodSequenceIndexSub$ and the provided index from the preceding step, i.e., $\indexOptimumPre$. It then calculates $\angleexpAllBitD[1]$  and $\permutationMonoD[1]$ by using $\setOfpermutationsSub$ and $\indexOptimumPre$ and combines them with the detected permutation and the phases from the preceding step. 

\subsubsection{De-mapping}
After backtracking is finalized, the detected $\seqStilde$ can be obtained from $\indexOptimum$. The parameter $\decimalForSepD$ can then be calculated as $\decimalForSepD =\algorithmDecoderDeltaSeparation[{\seqStilde},\remaningLength,\numberOfIterations,\spacing]-1$. Therefore, $\decimalForSepAndPermD$ can be obtained as
$\decimalForSepAndPermD = \decimalForSepD\numberOfIterations!/2+\decimalForPermD$, where $\decimalForPermD$ is the corresponding integer for $\seqPermutationCompShiftD$. 
It is worth noting that if the detected $\permutationMonoD[1]$ is less than the detected $\permutationMonoD[\numberOfIterations]$, the detector reverses the elements of $\seqPermutationCompShift$ and the detected $(\angleexpAll[1],\mydots,\angleexpAll[\numberOfIterations])$.
Finally, the decoder converts the decimal $\decimalForSepAndPermD$ and the sequence $(\angleexpAllBitD[1],\mydots,\angleexpAllBitD[\numberOfIterations], \arbitraryPhaseKBitD)$ to $\bitsForSepD$ and  $\bitsForPhaseD$, respectively. The pseudocode for the decoder is given in Algorithm~\ref{alg:decoder}. The receiver diagram is also provided in  \figurename~\ref{fig:txrx}\subref{subfig:rx}.

\subsection{Complexity}
At the preparation stage, the identification of  $\numberOfSelectSeqences$  best candidates for the separations based on \eqref{eq:score} requires the calculation of \eqref{eq:score} by $\cardinalityDistance[\numberOfIterations,\spacing][\remaningLength]$ times and an algorithm for identifying at most $\numberOfSelectSeqencesMax$ candidates. The time complexity for \eqref{eq:score} increases linearly with $2^\numberOfIterations$, $\cardinalityDistance[\numberOfIterations,\spacing][\remaningLength]$, and $\numberOfPointsForPSK$ and it is larger than the time complexity of the sorting algorithm. Therefore, the time complexity of the preparation stage is $\mathcal{O}(2^\numberOfIterations\numberOfPointsForPSK\cardinalityDistance[\numberOfIterations,\spacing][\remaningLength])$. It is worth noting that identifying the indices based on sorted subcarrier energy levels  can substantially reduce to $\mathcal{O}(2^\numberOfIterations\numberOfPointsForPSK)$. However, since a single misidentified index causes an {\em ordering} problem (i.e., the decoder needs to deal with not only noise but also permuted and/or erased elements), the decoder with such simplifications tends to work only at high \ac{SNR}. 

At the enumeration stage after the preparation, the decoder enumerates $\numberOfPointsForPSK(\numberOfIterations-1)\numberOfIterations\numberOfSelectSeqences$ sequences of length $2^{\numberOfIterations-1}$ from $\numberOfIterations\numberOfSelectSeqences$ sequences of length $2^{\numberOfIterations}$ based on \eqref{eq:newSeq}. The corresponding time complexity can be calculated as
$\mathcal{O}(2^{\numberOfIterations-1}\numberOfPointsForPSK\numberOfIterations\numberOfSelectSeqences)$. Based on \eqref{eq:scoreP}, the decoder then prunes a majority of the enumerated sequences, i.e.,  at most $\numberOfGoodSeqences$ of them survive. The time complexity of \eqref{eq:scoreP} can be calculated as $\mathcal{O}(2^{\numberOfIterations-1}\numberOfPointsForPSK\numberOfIterations\numberOfSelectSeqences)$. For the following steps, the decoder enumerates $\numberOfPointsForPSK(\numberOfIterations-\indexRecursion)\numberOfGoodSeqences$  sequences of length $2^{\numberOfIterations-\indexRecursion}$ at the $\indexRecursion$th recursion step for $2\le\indexRecursion\le\numberOfIterations-1$ and chooses $\numberOfGoodSeqences$ sequences. Therefore, the time complexity is $\mathcal{O}(2^{\numberOfIterations-\indexRecursion-1}\numberOfPointsForPSK\numberOfGoodSeqences)$ for the $\indexRecursion$th recursion step.

Based on our trials, we observe that $\numberOfSelectSeqences$ often needs to be larger than $\numberOfGoodSeqences$ for a large $\remaningLength$. Thus, the first enumeration step is the most expensive part as compared to the following steps in the recursions, which may prohibit the practical implementations for devices with low computational power. Hence, a lower-complexity decoding algorithm for the partitioned \acp{CS} is needed and it is currently an open problem.

\section{Numerical Results}
\label{sec:numeric}
In this section, we evaluate the partitioned \acp{CS} for an \ac{OFDM}-based communication system, numerically. For side-by-side comparisons, we consider the standard \acp{CS} proposed in \cite{davis_1999} and the polar code in 3GPP \ac{5G} \ac{NR} \ac{UL}. For the standard \acp{CS}, we consider $\numberOfIterations\in\{4,5,6,7\}$. For the partitioned \acp{CS}, we also consider the same sequence length as the standard \acp{CS}, i.e., $\lengthData\in\{16,32,64,128\}$, and set $\remaningLength$ as $\lengthData/2$ based on Figure~\ref{fig:cardinality}\subref{subfig:bitsOnPartAndNZ}. For the decoder, we set $\numberOfSelectSeqencesMax=10000$ and $\numberOfGoodSeqences=400$. The same decoder is used for the standard \acp{CS}. For all \acp{CS} configurations, we consider \ac{OFDM} and \ac{QPSK}, i.e., $\numberOfPointsForPSK=4$. 	  For a fair comparison, we adjust the code rate of the polar code such that it leads to the same \ac{SE} and codeword length as the \acp{CS}. We also add $11$ \ac{CRC} bits to the information bits and interleave the coded bits. For the decoder, we use a list decoder with soft information, where the list length is set to $64$. We use  \ac{DFT-s-OFDM} for polar code unless otherwise stated and consider both
 $\pi/2$-\ac{BPSK} and $\pi/4$-\ac{QPSK}, which are effective to mitigate the \ac{PMEPR} by reducing the zero-crossings in the time domain.	For the multi-path fading channel model, we use ITU Vehicular A  with no mobility. We assume that the \ac{CP} duration is longer than the maximum excess delay of the multi-path channel to avoid inter-symbol interference, the symbol duration is set to $\symbolDuration=66.67~\mu$s, and the transmitter and receiver are synchronized in both time and frequency.

\subsection{Minimum Distance, Spectral Efficiency, and Bandwidth}
\begin{figure}
	\centering
	\subfloat[Minimum distance versus \ac{SE}.]{\includegraphics[width =3.5in]{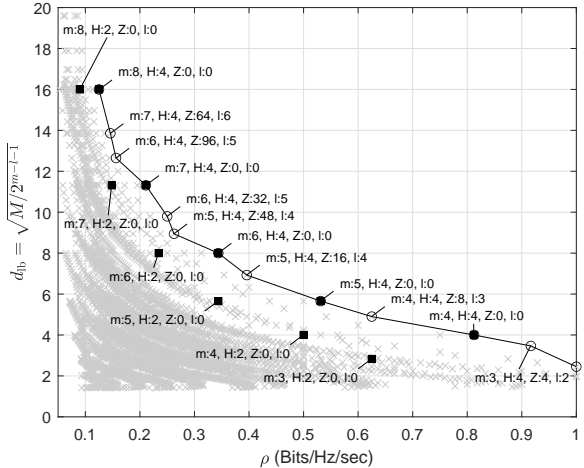}
		\label{subfig:dminVersusSE}}\\
	\subfloat[Minimum distance versus bandwidth for a given range of \ac{SE}.]{\includegraphics[width =3.5in]{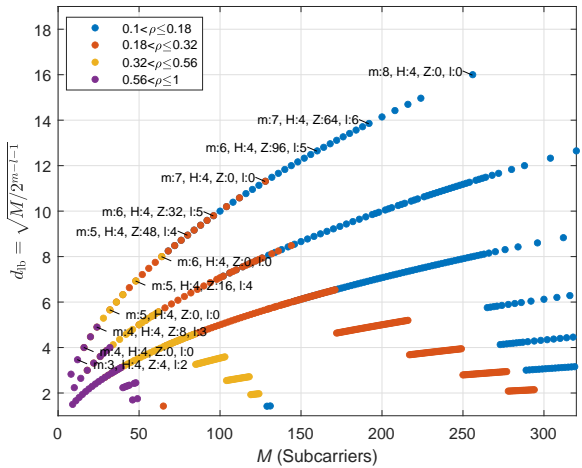}
		\label{subfig:dminVersusM}}
	\caption{The partitioned \acp{CS} can maintain the minimum distance properties of the standard \acp{CS} and be compatible with the non-power-two $\lengthData$.}
	\label{fig:dmin}
\end{figure}
In \figurename~\ref{fig:dmin}, we evaluate the trade-offs between minimum distance, \ac{SE}, and bandwidth for partitioned \acp{CS} and the standard \acp{CS} by calculating these parameters for  $\numberOfIterations\in\{3,4,5,6,7,8\}$, $0\le\spacing\le\numberOfIterations-1$, $\remaningLength\in\{0,1,\mydots,256\}$, and $\numberOfPointsForPSK=\{2,4\}$. Since $\minimumDistanceNZ=\sqrt{M}$ for $\numberOfPointsForPSK=\{2,4\}$, $\minimumDistance\ge\lowerBoundMinimumDistance\triangleq \sqrt{\lengthData/2^{\numberOfIterations-\spacing-1}}$ for $0\le\spacing\le\numberOfIterations-1$. In \figurename~\ref{fig:dmin}\subref{subfig:dminVersusSE}, we calculate  $\spectralEfficient$ and $\lowerBoundMinimumDistance$ for all combinations and mark several best cases that provide the maximum $\lowerBoundMinimumDistance$ for a given \ac{SE}. \figurename~\ref{fig:dmin}\subref{subfig:dminVersusSE} shows that the partitioned \acp{CS} support a wide range of \ac{SE} as opposed to the standard \acp{CS} based on \eqref{eq:imagPartReduced}. We also observe that the partitioned \acp{CS} with a smaller $\numberOfIterations$ can achieve a similar 
$\lowerBoundMinimumDistance$ of the standard \acp{CS}. For instance, $\lowerBoundMinimumDistance=8$ for $\numberOfIterations=6$ without partitioning, $\lowerBoundMinimumDistance=8.9$ for $\numberOfIterations=5$, $\remaningLength=48$, and $\spacing=4$. In \figurename~\ref{fig:dmin}\subref{subfig:dminVersusM}, we also take the impact of bandwidth on the minimum distance and \ac{SE} into account. For a given parameter set of \ac{SE} range and bandwidth, we plot the maximum  $\lowerBoundMinimumDistance$. For example, when $\lengthData=160$ subcarriers, the $\lowerBoundMinimumDistance=12.65$ is achieved for the configuration $\numberOfIterations=6$, $\remaningLength=96$, and $\spacing=5$, where the \ac{SE} is $\spectralEfficient=25/160$ bits/Hz/sec, which is between $0.1$ and $0.18$ bits/Hz/sec. %In another example, for the configuration $\numberOfIterations=3$, $\remaningLength=12$, and $\spacing=2$, the \ac{SE} is $\spectralEfficient=13/20$ bits/Hz/sec, i.e., within the range of $0.63$ and $0.79$ bits/Hz/sec, and it provides $\lowerBoundMinimumDistance=4.47$.
 \figurename~\ref{fig:dmin}\subref{subfig:dminVersusM} explicitly shows that a larger $\lowerBoundMinimumDistance$ is achieved for a larger $\lengthData$ or a smaller ${\spectralEfficient}$. Although the partitioning does not remedy the loss of \ac{SE} for a larger $\lengthData$  under a minimum distance constraint, it enables many different options for various \acp{SE} and bandwidths, which is beneficial for exploiting the available \ac{DoF} in the frequency domain better.

\subsection{PMEPR and Error Rate}

\begin{figure}[t]
	\centering
	{\includegraphics[width =3.5in]{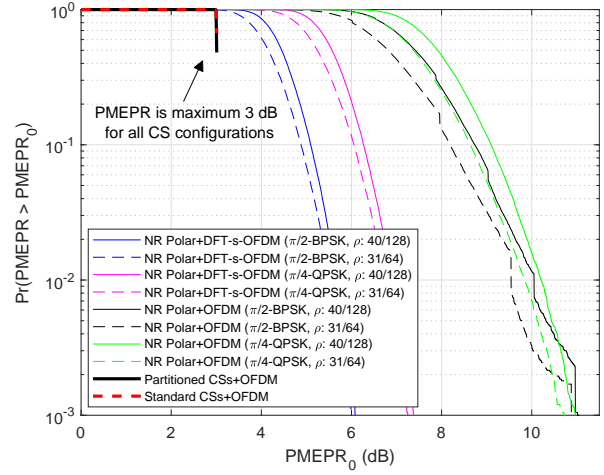}
	}
	\caption{{PMEPR distributions.}}
	\label{fig:PAPR}
\end{figure}
In \figurename~\ref{fig:PAPR}, we provide the \ac{PMEPR} distributions for various configurations, respectively. As expected, the partitioning maintains the \ac{PMEPR} benefit of the standard \acp{CS} as the partitioned sequences based on Theorem~\ref{th:reduced} are still \acp{CS}. Hence, the maximum \ac{PAPR} is  $3$~dB for all configurations of \acp{CS}. Since the partitioning does not alter the norm of the sequence, the mean \ac{OFDM} symbol power also remains constant. The polar code with \ac{DFT-s-OFDM} is superior to that of \ac{OFDM} in terms of \ac{PMEPR} as  \ac{DFT} pre-coding effectively converts the multi-carrier nature of the \ac{OFDM} to a form of single-carrier waveform. The \ac{PMEPR} distribution for the polar code with $\pi/2$-BPSK is better than the one with $\pi/4$-QPSK. This is expected because the zero-crossings in the time domain  are mitigated further with $\pi/2$-BPSK, as compared to ones with $\pi/4$-QPSK.  Nevertheless, the gap between \acp{CS} and the polar code with \ac{DFT-s-OFDM} with $\pi/2$-\ac{BPSK}  can still reach up to 3~dB.

\begin{figure*}
	\centering
	\subfloat[{AWGN channel.}]{\includegraphics[width =3.5in]{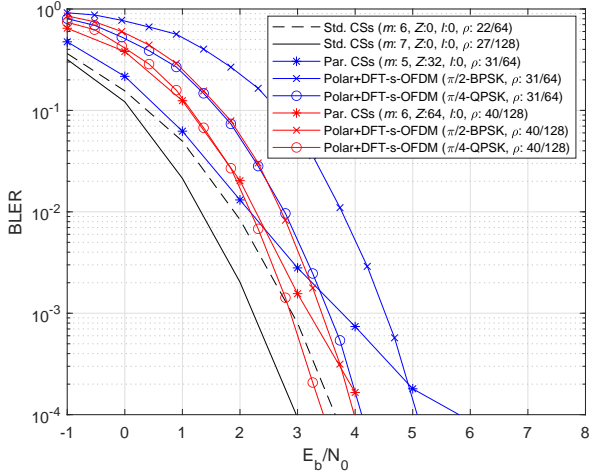}
		\label{subfig:blerebn0_AWGN_noDistance}}
	\subfloat[{Fading channel.}]{\includegraphics[width =3.5in]{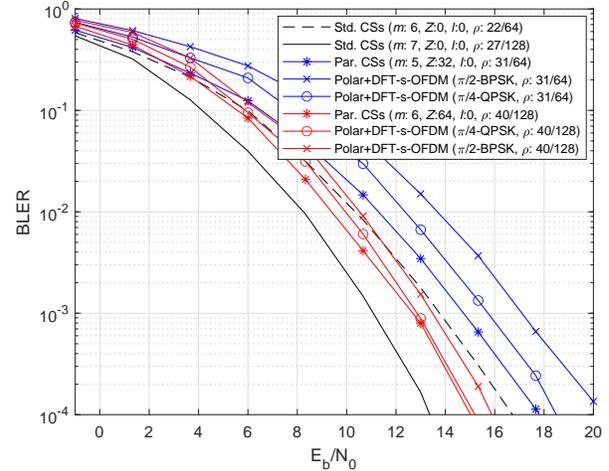}
		\label{subfig:blerebn0_fading_noDistance}}	
	\\
	\subfloat[AWGN channel ($\spacing>0$ for the partitioned CSs).]{\includegraphics[width =3.5in]{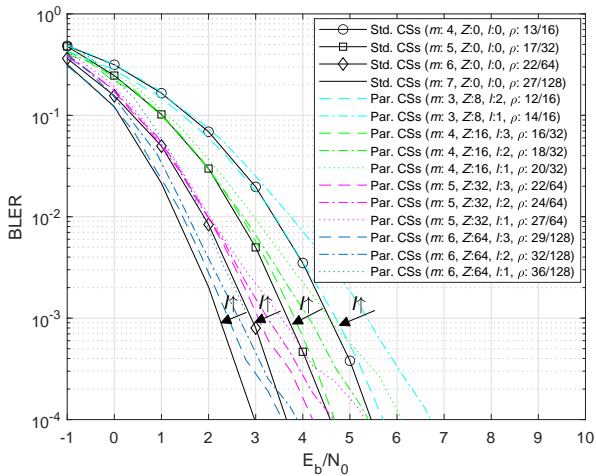}
		\label{subfig:blerebn0_AWGN}}	
	\subfloat[Fading channel ($\spacing>0$  for the partitioned CSs).]{\includegraphics[width =3.5in]{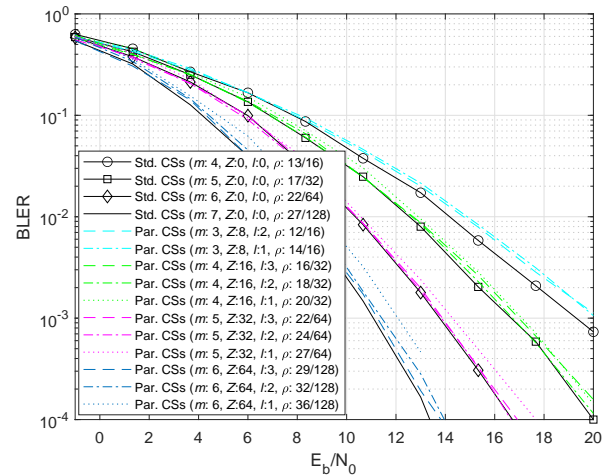}
		\label{subfig:blerebn0_fading}}
	\caption{BLER performance of the partitioned CSs as compared to that of the standard CSs in \cite{davis_1999} and polar code with \ac{DFT-s-OFDM}.}
	\label{fig:bler}
\end{figure*}
In \figurename~\ref{fig:bler}, we provide the BLER versus $\EbNO$ curves in \ac{AWGN} and fading channels. In \figurename~\ref{fig:bler}\subref{subfig:blerebn0_AWGN_noDistance} and \figurename~\ref{fig:bler}\subref{subfig:blerebn0_fading_noDistance}, we do not restrict the supports of the partitioned \acp{CS}, i.e., $\spacing=0$. As discussed in Section~\ref{sec:CSwithNS}, under the same bandwidth, the partitioned \acp{CS} can yield a higher \ac{SE} with a smaller $\numberOfIterations$  as compared to the standard \acp{CS}. For example, $\lengthData=128$, the number of information bits increases to $40$ for $\numberOfIterations=6$ whereas it is $27$ bits for the standard \acp{CS} with $\numberOfIterations=6$. However, since the partitioning  decreases the minimum Euclidean distance of the standard \acp{CS}, the \ac{BLER} increases for the partitioned \acp{CS}. At $1e-3$ \ac{BLER}, the losses in  $\EbNO$ are approximately $1$ dB for both  $\lengthData=64$ and $\lengthData=128$, respectively, for the \ac{AWGN} channel.
%Although the gap between the standard \acp{CS} and the partitioned \acp{CS} diminishes at $1e-3$ \ac{BLER} for a large $m$, 
The slopes of the curves for the standard \acp{CS} and the partitioned \acp{CS} are also different as the pairwise distance distribution of the set of partitioned \acp{CS} is different from the one for the standard \acp{CS}. 	We also observe that the case where $\remaningLength=64$ performs worse than the case  where $\remaningLength=32$ for low $\EbNO$. This is due to the  sequence identification at the preparation phase of the decoder, which requires a larger $\numberOfSelectSeqences$ for the case where $\remaningLength=64$ at the expense of higher complexity. In the fading channel, the $\EbNO$ losses are approximately $1$~dB at $1e-3$ \ac{BLER} for all cases. 

In \figurename~\ref{fig:bler}, we also analyze the error rate with the polar code with $\pi/2$-BPSK and $\pi/4$-QPSK under the same spectral efficiency provided by partitioned \acp{CS} . As expected, the error rate for the polar code with $\pi/4$-QPSK is better than the one with $\pi/2$-BPSK because a larger codeword length is utilized. Nevertheless, the performance of the polar code  with $\pi/4$-QPSK  is similar to that of the partitioned \acp{CS}. We observe that the performance differences are within the range of 1 dB at $1e-3$ \ac{BLER} in both \ac{AWGN} and fading channels. %The \acp{CS} are superior to the polar code for a large $\numberOfIterations$ in average \ac{BLER}.
While the polar code has an advantage in terms of receiver complexity as compared to the decoder for the partitioned \ac{CS} used in this study, the partitioned \ac{CS}-based encoding is superior to the polar code with $\pi/4$-QPSK   in terms of \ac{PMEPR} as shown in \figurename~\ref{fig:PAPR}. Note that the slope of the error rate curves in the fading channel is a function of the bandwidth for all schemes. This is because more diversity gain is achieved for a larger $\lengthData$.

In \figurename~\ref{fig:bler}\subref{subfig:blerebn0_AWGN} and \figurename~\ref{fig:bler}\subref{subfig:blerebn0_fading}, we consider the cases where the supports of the partitioned \acp{CS} are restricted to achieve a larger minimum distance. The $\EbNO$ loss quickly diminishes for all cases with an increasing $\spacing>0$ in both \ac{AWGN} and fading channel. For a large $\spacing$, the partitioned \acp{CS} perform similar to the standard \acp{CS} with a similar \ac{SE}. For example, for the case $(\numberOfIterations=6,\remaningLength=64,\spacing=3)$, the \ac{SE} is $29/128$ bits/Hz/sec and it performs similar to the standard \acp{CS}, where $\numberOfIterations=7$ and the corresponding \ac{SE} is $27/128$  bits/Hz/sec. The minimum distance  for the standard \acp{CS} can be calculated as $\sqrt{128}$ when $\numberOfIterations=7$. However, $\minimumDistance\ge\lowerBoundMinimumDistance=\sqrt{32}$ for the partitioned \acp{CS}. 
The diminishing loss in the average \ac{BLER} implies that a majority of the pairwise distances for the partitioned \acp{CS} are actually larger than $\lowerBoundMinimumDistance=\sqrt{32}$. Similar observations can also be made for other cases.

\section{Conclusion}
In this study, we  analyze partitioned \acp{CS} for \ac{OFDM}. We analytically obtain the number of partitioned \acp{CS} for given bandwidth and a minimum distance constraint. We derive the corresponding recursive methods for mapping a natural number to a separation and vice versa. We show that the partitioning rule under Theorem~\ref{th:reduced} has symmetrical characteristics, which can be utilized to restrict the partitions based on a minimum distance constraint. In addition, we develop an encoder and a recursive decoder for partitioned \acp{CS}. 

Our results indicate that a larger number of \acp{CS} can be synthesized  as compared to the standard \acp{CS} for a given bandwidth through partitioning. For example, for $512$ \ac{DoF}, we show that the number of distinct \acp{CS} is increased approximately by a factor of $2^{25}$  without changing the alphabet of the non-zero elements of the \ac{CS} (i.e., the standard sequences with $\numberOfPointsForPSK$-\ac{PSK}) with the partitioning. Hence, partitioning provides a way of achieving a larger \ac{SE} without changing the alphabet of the non-zero elements of the \acp{CS} or using a high-order modulation for \acp{CS}. Partitioning without any restriction decreases the minimum Euclidean distance of the standard \acp{CS}. By using our encoder and detector, we show that the \ac{SNR} losses are around $2$ dB and $1$ dB in \ac{AWGN} and fading channels at $1e-3$ \ac{BLER}, respectively, as compared to the standard \acp{CS}. However, under a minimum distance constraint, the partitioned \acp{CS} perform similar to the standard \acp{CS} in terms of error average \ac{BLER} and they can utilize the available resources in the frequency domain better than the standard \acp{CS}. 

The partitioned \acp{CS} that we investigate in this paper extend the standard \acp{CS} from Davis and Jedwab's encoder in \cite{davis_1999} and have some desirable properties such as the adjustable minimum Euclidean distance. To fully exploit the properties of partitioned \acp{CS}, a low-complexity decoder  is  needed, which is currently a difficult open problem. Extending the non-zero elements of partitioned \acp{CS} to a larger constellation  to increase the minimum distance is also another future research direction that can be pursued. Because of the absence of a tighter  \ac{SE} bound that takes the sequence length, \ac{SNR}, and \ac{PMEPR} constraints than the available ones, e.g. \cite{Paterson_2000bound, Tarokh_2000bound}, in this study, we cannot compare the obtained \ac{SE} of the proposed scheme with a theoretical bound. Hence, a tight upper bound on the \ac{SE} is needed for a better understanding of the theoretical limits under \ac{PMEPR} constraints.

\label{sec:conc}
\bibliographystyle{IEEEtran}
\bibliography{csWithNonSquashingIntegers}

\end{document}

%% file: variables.tex
\def\complexNumbers{\mathbb{C}}
\def\realNumbers{\mathbb{R}}
\def\integers{\mathbb{Z}}
\def\integersPositive{\mathbb{Z}^{+}}
\def\integersNonnegative{\mathbb{Z}_{0}^{+}}

\def\constante{{\rm e}}
\def\constantj{{\rm j}}
\def\exponentialBase{\xi}

\def\numberOfEnumarationSequences[#1]{\mathcal{D}^{(#1)}(M,H)}

\def\lagForCorrelation{k}
\def\unitSequences{\mathcal{C}^{(m)}(H)}

\def\indexEleOfSeq{i}
\def\indexIteration{n}

\def\indexIterationANF{l}

\def\indexFirstOrderMonomial{j}
\def\indexMonomial{k}
\def\orderMonomial[#1]{k_{#1}}
\def\coeffientsANF[#1]{c_{#1}}
\def\polyVariable{z}

\def\numberOfIterations{m}
\def\numberOfPointsForPSK{H}
\def\numberOfPointsForPSKexp{h}
\def\modulationSymbolF[#1]{m_{#1}}
\def\cardinalitySetOfOperators[#1]{{H}_{#1}}

\def\varMonomial{{i_{\rm d}}}
\def\monomial[#1]{x_{#1}}

\def\lengthData{M}

\def\scaleAexp[#1]{a_{#1}}
\def\scaleBexp[#1]{b_{#1}}
\def\scaleEexp[#1]{e_{#1}}
\def\angleexp[#1]{c_{#1}}
\def\angleexpAll[#1]{k_{#1}}
\def\angleScaleAexp[#1]{\dot{c}_{#1}}
\def\angleScaleBexp[#1]{\ddot{c}_{#1}}

\def\arbitraryPhaseK{k'}

\def\separationGolay[#1]{d_{#1}}
\def\separationGolayFix{d'}

\def\angleexpAllBitD[#1]{\hat{k}_{#1}}
\def\arbitraryPhaseKBitD{\hat{k}'}

\def\permutationMonoD[#1]{{\hat{\pi}_{#1}}}
\def\seqPermutationCompShiftD{\bm{\hat{\pi}}}
\def\eleGa[#1]{{a}_{#1}}
\def\eleGb[#1]{{b}_{#1}}
\def\eleGc[#1]{{c}_{#1}}
\def\eleGt[#1]{{t}_{#1}}

\def\apac[#1][#2]{\rho_{#1}(#2)}
\def\apacPositive[#1][#2]{\rho^{+}_{#1}(#2)}
\def\binaryAsignment[#1][#2]{b_{#1}^{(#2)}}

\def\eleSeqf[#1]{{f}_{#1}}
\def\eleSeqg[#1]{{g}_{#1}}
\def\eleSeqcf[#1]{{c}_{f,#1}}
\def\eleSeqcg[#1]{{c}_{g,#1}}

\def\scaleA[#1]{\alpha_{#1}}
\def\scaleB[#1]{\beta_{#1}}
\def\angleGolay[#1]{\omega_{#1}}
\def\angleScaleA[#1][#2]{\dot{\omega}_{#1}^{#2}}
\def\angleScaleB[#1][#2]{\ddot{\omega}_{#1}^{#2}}
\def\permutationShift[#1]{{\psi_{#1}}}

\def\permutationMono[#1]{{\pi_{#1}}}
\def\permutationMonoPre[#1]{{\pi'_{#1}}}

\def\seqPermutationCompShift{\bm{\pi}}
\def\seqPermutationCompShiftSub{\bm{\pi}'}

\def\symbolDuration{T_{\rm s}}

\def\timeVar{t}

\def\Ptransmit{P_{\rm tx}}

\def\separationIterative[#1]{{\tau}_{#1}}
\def\permutationShiftAux[#1]{{\phi_{#1}}}
\def\separationGolayAux[#1]{\tau'_{#1}}

\def\funczArbitrary[#1]{z(#1)}
\def\funcaArbitrary[#1]{a(#1)}
\def\funcbArbitrary[#1]{b(#1)}
\def\coefficientArbitrary[#1]{k_{#1}}

\def\distanceToPoint[#1]{d_{#1}}
\def\ratioBetweenDistanceAndInner[#1]{d_{#1}} % change it from l to d
\def\angleBetweenPointAndXaxis[#1]{\theta_{#1}}
\def\angleBetweenPointAndXYDiagonal[#1]{\psi_{#1}}
\def\carrierFrequency{f_{\rm c}}

\def\vecArrangement[#1]{\textbf{b}_{#1}}

\def\seqGa{\textit{\textbf{a}}}
\def\seqGb{\textit{\textbf{b}}}
\def\seqGaIt[#1]{\textit{\textbf{a}}^{(#1)}}
\def\seqGbIt[#1]{\textit{\textbf{b}}^{(#1)}}
\def\seqGc{\textit{\textbf{c}}}

\def\seqGf[#1]{\textit{\textbf{f}}_{#1}}
\def\seqGg[#1]{\textit{\textbf{g}}_{#1}}
\def\seqGfdot[#1]{\bar{\textit{\textbf{f}}}_{#1}}
\def\seqGgdot[#1]{\bar{\textit{\textbf{g}}}_{#1}}

\def\OFDMinTime[#1][#2]{s_{#1}(#2)}

\def\seqGaP{A}

\def\seqGaItP[#1]{{A}^{(#1)}}
\def\seqGbItP[#1]{{B}^{(#1)}}
\def\seqGcP{C}

\def\seqSubP[#1]{{H}_{#1}}
\def\seqGfP[#1]{{{F}}_{#1}}
\def\seqGgP[#1]{{{G}}_{#1}}

\def\seqGt{\textit{\textbf{t}}}
\def\seqGtP{{{T}}}

\def\seqSub[#1]{\textit{\textbf{h}}_{#1}}

\def\seqFirstOrderMonomial[#1]{\textit{\textbf{m}}_{#1}}
\def\seqx{\textit{\textbf{x}}}

\def\seqToBeModulated[#1]{\textit{\textbf{s}}_{#1}}

\def\flipConjugate[#1]{{{\tilde{#1}}}}
\def\expectationOperator[#1][#2]{{\mathbb{E}_{#2}}[#1]}
\def\operator[#1][#2]{\mathcal{O}_{#1}^{(#2)}}
\def\operatordot[#1][#2]{\bar{\mathcal{O}}_{#1}^{(#2)}}

\def\compositeOperatorF[#1][#2]{{F}_{#1}{(#2)}}
\def\compositeOperatorG[#1][#2]{{G}_{#1}{(#2)}}
\def\compositeOperatorFdot[#1][#2]{\bar{F}_{#1}{(#2)}}
\def\compositeOperatorGdot[#1][#2]{\bar{G}_{#1}{(#2)}}

\def\setOfOperators[#1]{{\mathfrak{J}}_{#1}}

\def\operatorBinary[#1][#2]{O_{#1}^{(#2)}}
\def\operatorSign[#1][#2]{{\rm S}_{#1}^{(#2)}}
\def\operatorScaleA[#1][#2]{{\rm{A}}_{#1}^{(#2)}}
\def\operatorScaleB[#1][#2]{{\rm{B}}_{#1}^{(#2)}}
\def\operatorAngle[#1][#2]{\Omega_{#1}^{(#2)}}
\def\operatorSeparation[#1][#2]{\Delta_{#1}^{(#2)}}
\def\operatorOrderA[#1][#2]{\dot{\rm O}_{#1}^{(#2)}}
\def\operatorOrderB[#1][#2]{\ddot{\rm O}_{#1}^{(#2)}}
\def\operatorAngleScaleA[#1][#2]{\dot{\Omega}_{#1}^{(#2)}}
\def\operatorAngleScaleB[#1][#2]{\ddot{\Omega}_{#1}^{(#2)}}
\def\operatorAngleConjScaleA[#1][#2]{\dot{\Upsilon}_{#1}^{(#2)}}
\def\operatorAngleConjScaleB[#1][#2]{\ddot{\Upsilon}_{#1}^{(#2)}}

\def\functionf[#1]{p^{(#1)}}
\def\functiong[#1]{q^{(#1)}}

\def\functionfdot[#1]{{p}_{\indexIterationANF}^{(#1)}} %removed bars
\def\functiongdot[#1]{{q}_{\indexIterationANF}^{(#1)}} % removed bars

\def\funcfForCommonShift{f_{\rm s}}
\def\funcfForCommonShiftDec{\check{f}_{\rm s}}

\def\funcfForFinalPhase{f_{\rm i}}
\def\funcfForFinalPhaseDec{\check{f}_{\rm i}}

\def\funcfForANF{f}

\def\funcfForANFdec{\check{f}}
\def\funcEnum{{e}}
\def\funcEnumInv{{e^{-1}}}

\def\varMonomial{{i}}

\def\funcGfForANF[#1]{f_{#1}}
\def\funcGgForANF[#1]{g_{#1}}
\def\polySeq[#1][#2]{{#1}(#2)}

\def\shortMinus{\scalebox{0.75}[1.0]{\( - \)}}
\def\spectralEfficient{\rho}
\def\setMonomials[#1]{{\mathcal{M}_{#1}}}

\def\phaseOffsetl[#1]{\Delta_{#1}}
\def\phaseOffsetll[#1]{\Delta_{#1}}
\def\amplitudeOffset[#1]{\epsilon'_{#1}}
\def\amplitudeOffsetl[#1]{\epsilon_{#1}}

\def\referenceDerivative{\ell}

\def\receivedElement[#1]{r_{#1}}
\def\channel[#1]{c_{#1}}
\def\noise[#1]{n_{#1}}

\def\parametersImag{\theta_{\rm i}}

\def\parametersImagg{{{\psi}}_{\rm i}}

\def\parametersImagf{{{\phi}}_{\rm i}}

\def\weightedElement[#1]{s_{#1}}
\def\channelPowerElement[#1]{h_{#1}}
\def\weightedElementAfterSum[#1][#2]{s_{#1}^{#2}}
\def\channelPowerElementAfterSum[#1][#2]{h_{#1}^{#2}}

\def\funcfForFinalAmplitudeDerivate[#1][#2]{f_{\rm r}^{#1}(#2)}
\def\funcfForFinalPhaseDerivate[#1][#2]{f_{\rm i}^{#1}(#2)}
\def\funcgForFinalAmplitudeDerivate[#1][#2]{g_{\rm r}^{#1}(#2)}
\def\funcgForFinalPhaseDerivate[#1][#2]{g_{\rm i}^{#1}(#2)}
\def\mappingFunction[#1]{{I_{\referenceDerivative}(#1)}}
\def\algorithmDecoder{A}

\def\weightedSequence[#1]{\textit{{\textbf{s}}}_{#1}}
\def\weightedChannel[#1]{\textit{{\textbf{h}}}_{#1}}
\def\numberOfSequences{L}

\def\setOfweightedSequence{\textit{\textbf{S}}_{\numberOfIterations}}
\def\setOfweightedChannel{\textit{\textbf{H}}_{\numberOfIterations}}
\def\setOfpermutations{\textit{\textbf{L}}_{\numberOfIterations}}

\def\setOfweightedSequenceSub{\textit{\textbf{S}}_{\numberOfIterations-1}}
\def\setOfweightedChannelSub{\textit{\textbf{H}}_{\numberOfIterations-1}}
\def\setOfpermutationsSub{\textit{\textbf{L}}_{\numberOfIterations-1}}

\def\lengthOfSequence{L}

\def\indexOptimum{\tilde{n}}

\def\indexEnumaration{n}

\def\numberOfGoodSeqences{{N_{\rm best}}}

\def\EbNO{E_{\rm b}/N_{0}}

\def\indexRecursion{i}